\documentclass[copyright,creativecommons]{eptcs}
\usepackage{breakurl}
\providecommand{\event}{PLACES 2010} 

\newcommand{\qedhere}{\nobreak\ifvmode\relax\else%
        \ifdim\lastskip<1.5em \hskip-\lastskip%
        \hskip 1.5em plus0em minus0.5em%
        \fi\nobreak\strut\hfill\ensuremath{\square}%
        \fi}

\newenvironment{proof}[1][\bfseries{Proof.}]{\begin{trivlist}
  \item[\hskip \labelsep {\bfseries{#1}}]}{\qedhere\end{trivlist}}

\title{%
   A Type System for Unstructured Locking
   that Guarantees Deadlock Freedom
   without Imposing a Lock Ordering
}
\def\titlerunning{Type-Based Deadlock Freedom without Lock Ordering}

\usepackage{color}

\newif\ifeptcstr
\newif\ifeptcs
\newif\ifplaces
\newif\ifnotplaces
\newif\ifllncs
\newif\iftypesafety
\newif\ifplacestr
\eptcstrue
\notplacestrue
\placestrtrue

\usepackage{cancel}
\usepackage{xspace}

\usepackage{amsmath}
\usepackage{amstext}
\usepackage{amssymb}
\usepackage{amsbsy}
\usepackage[inference]{semantic}
\usepackage{txfonts}

\usepackage{ndisplay}
\usepackage{ngrammar}
\usepackage{color}

\newcommand\cfont[1]{\ensuremath{\mathtt{#1}}}

\makeatletter
\newenvironment{nruledisplay}{%
  \begin{center}%
  \(%
  \begin{array}[b]{@{}c@{}}%
}{%
  \end{array}%
  \)%
  \end{center}%
}

\def\nrule{\@ifnextchar[{\nrule@label}{\nrule@simple}}
\def\nrulelabel#1{\mbox{\textsl{#1}}}
\def\nrule@label[#1]#2#3{\ensuremath{%
  \begin{array}[t]{@{}l@{}}%
  \inference{#2}{#3}[\text{\textsl{(#1)}}]%
  \end{array}%
}}
\def\nrule@simple#1#2{\ensuremath{%
  \begin{array}[t]{@{}l@{}}%
  \inference{#1}{#2}%
  \end{array}%
}}
\makeatother

   \newif\ifoperational
   \newif\ifnotplaces
   \ifplaces
      \newcommand{\rline}{\ensuremath{\\[1.5em]}\footnotesize}
      \newcommand{\rspace}{\ensuremath{\hspace{2em}}}
      \newcommand{\mfun}[2]{\mathsf{#1}(#2)}
      \newcommand{\loc} {{\rho}}
      \operationalfalse
      \newcommand{\cstd}{}
   \else
      \newcommand{\rline}{\\[2em]\footnotesize}
      \newcommand{\rspace}{\ensuremath{\hspace{2em}}}
      \newcommand{\mfun}[2]{\mathsf{#1}(#2)}
      \newcommand{\loc} {{r}}
      \operationaltrue
      \notplacestrue
      \newcommand{\cstd}{\ensuremath{M;}}
   \fi %

\newcommand{\mycomment}[1]{}

   \newcommand{\jjmath}{\ensuremath{\cfont{j}}}

\newcommand\highlight[2]{{\setlength\fboxsep{1pt}\colorbox{#1}{#2}}}
\def\PG{\highlight{colorPG}}
\definecolor{colorPG}{rgb}{1.0,0.8,0.8}

\definecolor{colorNSP}{rgb}{0.8,1.0,0.8}

\definecolor{colorKS}{rgb}{0.8,0.8,1.0}

      \makeatother
   \newlength{\gmylen}

   \newdimen\argwidth
      \def\[[#1\]]{%
         \setbox0=\hbox{$#1$}\argwidth=\wd0
            \setbox0=\hbox{$\left[\box0\right]$}\advance\argwidth by -\wd0
            \left[\kern.3\argwidth\box0\kern.3\argwidth\right]}

      \makeatletter
      \renewcommand\ngrammar@plain@rule@symbol{::=}%
      \renewcommand\ngrammar@plain@nst[1]{%
         \ngrammar@space%
            \ensuremath{\cfont{#1}}%
            \ngrammar@space@neededtrue}%
            \renewcommand\ngrammar@plain@nsnt[1]{%
               \ngrammar@space%
                  \ensuremath{\mathit{#1}}%
                  \ngrammar@space@neededtrue}%
                  \newcommand\nscomment[1]{%
                     \ifngrammar@long%
                        \iftypesafety
                        \nabstab{-\ngrammar@sem@width}%
                        \else
                        \ntab
                        \fi
                        \textit{#1}%
                     \else%
                          \ensuremath{& \textit{#1}}%
                     \fi%
                  }

   \newcommand{\gedef}[3]{
      \global\advance\ngrammar@numrulescount 1\relax%
         \ngrammar@numrule{\the\ngrammar@numrulescount}%
         \textbf{#1} \nabstab{\gmylen}
      \ifngrammar@long%
         \ndefinebegin{\nsnt{#2}}{\ngrammar@rule@symbol}%
         \else%
         & \nsnt{#2} & \ngrammar@rule@symbol &
         \fi%
         \ngrammar@activetrue%
         \ngrammar@space@neededfalse%
#3%
         \ifngrammar@long%
         \ndefineend%
         \fi%
         \ngrammar@activefalse%
         \\[\ngrammarruleskip]%
   }

   \newcommand{\fspan}[3]{
      \mfun{frame\_lockset}{#1,#2,\, #3}
   }

   \newcommand{\fmax}[1]{
      \mfun{max}{#1}
   }
   \newcommand{\fmin}[1]{
      \mfun{min}{#1} 
   }

   \newcommand{\fsub}[4]{
      #1 \,+_{#2} \, (#3,#4)
   }

   \newcommand{\fsplit}[4]{
     (#1,#2) =  \mfun{split}{#3 \, , \ #4} 
   }

   \newcommand{\flockset}[3]{
      \mfun{lockset}{#1,#2,#3}
   }

   \ifplacestr
      \ifllncs
      \else
         \newtheorem{lemma}{Lemma}
         \newtheorem{theorem}{Theorem}
      \fi
      \newtheorem{defn}{Definition}
      \newtheorem{theorema}{Theorem}
      \newtheorem{lemmaa}{Lemma}
   \else
      \iftypesafety

      \fi
   \fi
      {\begin{defn}[#1]\strut\rm}%
   {\end{defn}}

   \newcommand{\gbox}[3]{%
            \global\setlength{\gmylen}{#1}
            \footnotesize
               \begin{ngrammar}[noindent,semwidth=#2,long]
#3
               \end{ngrammar}
            \vspace*{-\baselineskip}
            \vspace*{-\ngrammarruleskip}
            \footnotesize
     }

         \newcommand{\opbox}{\ensuremath{
               \Box
            }}

         \newcommand{\aequiv}{\ensuremath{\equiv}}
         
         \newcommand{\stla}[2]{\ensuremath{\overset{#1}{{\longrightarrow}}}}

         \newcommand\oapp[3]{\ensuremath{%
            {\cfont{(}#1 \ \, #2 \cfont{)}}^{#3} }}

         \newcommand\oderef[1]{\ensuremath{%
            \cfont{deref} \ #1}}

         \newcommand\oassign[2]{\ensuremath{%
            #1 \ \cfont{:=} \ #2}}

         \newcommand{\oth}[2]
            {#1 \, \cfont{:} \, #2}

         \newcommand\olet[5]{\ensuremath{%
            \cfont{let} \ #2, #3 \ \cfont{= } \ \cfont{ref} \ #4 \ \cfont{in} \ #5}}

         \newcommand\oloc[1]{\ensuremath{%
            \cfont{loc}_{#1}}}
        
         \newcommand\oshare[2]{\ensuremath{%
            \cfont{share} \ {#2}}}

         \newcommand\orelease[2]{\ensuremath{%
            \cfont{release} \ {#2}}}

         \newcommand{\ocap}[3]{\ensuremath{%
            \PG{ocap FIXME!}
         }}

         \newcommand\olock[2]{\ensuremath{%
            \cfont{lock}_{#1} \ {#2}}}

         \newcommand\ounlock[2]{\ensuremath{%
            \cfont{unlock} \ {#2}}}

         \newcommand\ofunc[4]{\ensuremath{%
            \lambda #2.\, #3 \ \cfont{as}\ #4}}

         \newcommand\opoly[2]{\ensuremath{%
            \Lambda #1.\, #2}}

         \newcommand\orapp[2]{\ensuremath{%
            (#1)\, [#2]}}

         \newcommand\ounit{\ensuremath{%
            \cfont{()}}}

         \newcommand\oeval[2]{\ensuremath{%
            \cfont{pop}_{#1} \ #2}}

         \newcommand\nL[1]{\ensuremath{\mathsf{seq(}#1\mathsf{)}}\xspace}
         \newcommand\nP[1]{\ensuremath{\mathsf{par}}\xspace}

            \newcommand{\teff}[2]{\ensuremath{\,\&\, \cfont{(}#1\cfont{;}#2\cfont{)}}}

         \newcommand\tetyp[4]{\ensuremath{#1\vdash_{#4} #2 : #3}}

         \newcommand\ttype[5]{\ensuremath{
            \tetyp{#1}{#2}{#3 \teff{#4}{#5}}{}
         }}

         \newcommand\tetypext[4]{\ensuremath{#1\vdash_{#4} #2 : #3}}
         \newcommand\ttypeext[5]{\ensuremath{
            \tetypext{#1}{#2}{#3 \teff{#4}{#5}}{t}
         }}

         \newcommand\tholds[2]{\ensuremath{#1  \vdash #2}}
            \newcommand\teholds[3]{\ensuremath{#1  \vdash_{#3} #2}}
         
             \newcommand\ttuple[1]{\ensuremath{%
                                    \langle#1\rangle}}
             \newcommand\tunit{\ensuremath{%
                                 \ttuple{}}}
             
             \newcommand\tref[2]{\ensuremath{%
                                 \cfont{ref}(#1, #2)}}
             \newcommand\tfunc[4]{\ensuremath{%
                                    #1 \, \stla{#2}{#3}\, #4 }}
              \newcommand\tforall[2]{\ensuremath{%
                                       \forall \cfont{#1}.\,#2}}

         \newcommand{\stela}[2]{\ensuremath{\overset{#1;#2}{{%
                                 \longrightarrow}}}}
             \newcommand\tefunc[4]{\ensuremath{%
                                    #1 \, \stela{#2}{#3}\, #4 }}

\newcommand{\cstda}{\ensuremath{\cstd\Delta;\Gamma}}
\newcommand{\cstdD}[1] {\ensuremath {\cstd\Delta,#1;\Gamma} }
\newcommand{\cstdDG}[2] {\ensuremath{\cstd\Delta,#1;\Gamma,#2} }
\newcommand{\cT}{\ensuremath{\cstd\Delta} }

\newcommand\cprog {\ensuremath {S}}

\newcommand\conelt[3]{\ensuremath{%
   {#1}^{#2}}}

\newcommand\tbool{\ensuremath{%
\cfont{bool}%
}}
\newcommand\otrue{\ensuremath{%
   \cfont{true} }}
\newcommand\ofalse{\ensuremath{%
   \cfont{false} }}

\newcommand\oif[3]{\ensuremath{%
 \cfont{if} \,\, {#1} \,\, \cfont{then} \,\, {#2} \,\, %
  \cfont{else} \,\, {#3}}}

\newcommand\ofix[3]{\ensuremath{%
 \cfont{fix} \ {#1}:{#2}. \, {#3} }}

\newcommand{\tfsep}{\ensuremath{%
\,\cfont{?}\,%
}}

\newcommand{\trIF}[1][\relax]{
\ifx#1\relax
   \nrule[T-IF]{
       \cstda \vdash e_1 : \tbool 
         \teff{\gamma, \gamma_2 \tfsep \gamma_3}{\gamma'}
       &
        \fmax{\gamma::\gamma_2} =
         \fmax{\gamma::\gamma_3}
       \\
       \cstda \vdash e_2 : 
       \tau \teff{\gamma}{\gamma :: \gamma_2}
       &
       \cstda \vdash e_3 : 
       \tau \teff{\gamma}{\gamma :: \gamma_3}
   }{
    \cstda \vdash \oif{e_1}{e_2}{e_3} :
    \tau \teff{\gamma}{\gamma'}
   }
\else
   \nrule[T-IF]{
       \cstda \vdash e_1 : \tbool 
         \teff{\gamma, \gamma_2 \tfsep \gamma_3}{\gamma'}
       &
        \fmax{\gamma::\gamma_2} =
         \fmax{\gamma::\gamma_3}
       \\
       \cstda \vdash e_2 : 
       \tau \teff{\gamma}{\gamma :: \gamma_2}
       &
       \cstda \vdash e_3 : 
       \tau \teff{\gamma}{\gamma :: \gamma_3}
   }{
    \cstda \vdash \oif{e_1}{e_2}{e_3} :
    \tau \teff{\gamma}{\gamma'}
   }
\fi
}

\newcommand{\trFX}{
   \nrule[T-FX]{
      \tau \aequiv \tfunc{\tau_1}{\gamma_b}{}{\tau_2}
      &
      \tau' \aequiv \tfunc{\tau_1'}{\gamma_a'}{}{\tau_2'}
      &
      \tau \simeq \tau'
      &
      \gamma_a \simeq \gamma_a'
      \\
      \cstda,x:\tau \vdash f :  \tau' \teff{\gamma}{\gamma}
      &
      \gamma_b = \mfun{summary}{\gamma_a}
   }{
     \cstda \vdash   \ofix{x}{\tau}{f} : \tau \teff{\gamma}{\gamma}
   }
}
\newcommand{\vvdash}{\vdash_t}

\newcommand{\grRegion}{
   \gedef{Location}{r}{
      \rho 
      \nsor 
      \imath @n
      \nsor
      \rho @ n
   }
}

\newcommand{\grScope}{
   \gedef{Calling mode}{\xi}{
      \nL{\gamma} \nsor \nP{\pi}}
}

\newcommand{\grCap}{
   \gedef{Capability}{\kappa}{ n,n \nsor \overline{n,n}  }
}

\newcommand{\grEffectElt}{
   \gedef{Effect}{\gamma}{
      \emptyset 
         \nsor 
         \gamma,
         \ifoperational
            r^\kappa
         \else
            \rho^\kappa
         \fi
         \nsor
         \gamma, 
         \gamma 
         \tfsep 
         \gamma
   }
}

\newcommand{\grType}{
   \gedef{Type}{\tau}{
      \tunit \nsor
         \tfunc{\tau}{\gamma}{\gamma}{\tau} 
      \nsor
      \tforall{\rho}{\tau} 
      \nsorl
      \tref{\tau}{\loc}
      \nsor
      \tbool
   }
}

\newcommand{\grFunc}{
   \gedef{Function}{f}{
      \ofunc{\gamma}{x}{e}{
         \tfunc{\tau}{\gamma}{\gamma}{\tau}
      }
      \nsor
         \opoly{\rho }{f}
      \nsor
      \ofix{x}{\tau}{f}
   }
}

\newcommand{\grVal}{
   \gedef{Value}{v}{
      f \nsor
         \ounit \nsor
         \oloc{\imath} 
         \nsor
         \otrue
         \nsor
         \ofalse
   }
}

\newcommand{\grExpr}{
   \gedef{Expression}{e}{
      x  \nsor
      f \nsor
        \oapp{e}{e}{\xi} \nsor
         \orapp{e}{\loc}
      \nsor
         \oassign{e}{e} 
      \nsorl
      \oderef{e}
       \nsor
         \olet{\gamma}{\rho }{x}{e}{e} 
      \nsorl
         \oshare{\gamma}{e} 
      \nsor 
         \orelease{\gamma}{e}
      \nsor
         \olock{\gamma}{e}
      \nsorl
         \ounlock{\gamma}{e}
      \nsor
      \ounit
      \ifoperational
         \nsor
         \oeval{\gamma}{e}
         \nsor
         \oloc{\imath} 
      \fi
      \nsorl
                     \oif{e}{e}{e}
                     \nsor
                     \otrue
                     \nsor
                     \ofalse
   }
}

\newcommand{\oparrowt}[1]{ \, \ensuremath{ \leadsto }} 

\newcommand{\oparrowts}[1]{ \, \ensuremath{ \leadsto^{#1} }} 

\newcommand{\grDynCounts}{
      \gedef{Access Lists}{\theta}{ 
         \emptyset \nsor
            \theta,\imath \mapsto n;n;\epsilon;\epsilon
      }
}

\newcommand{\grRList}{
   \gedef{Store}{S}{  \emptyset \nsor S,\imath \mapsto v
   }
}

\newcommand{\grThread}{
   \gedef{Threads}{T}{\emptyset \nsor T, \oth{n}{\theta;e}   }
}

\newcommand{\grConf}{
   \gedef{Configuration}{C}{S;T}
}

\newcommand{\grEpsilon}{
   \gedef{Locations}{\epsilon}{\emptyset \nsor \epsilon,\imath}
}

\newcommand{\grEvalCont}[1][\relax]{\ensuremath{
   \gedef{Stack}{E}{
         \opbox \nsor 
         E[F]
      }
   \gedef
   {Frame}{F}{
         \oapp{\opbox}{e}{\xi} 
         \nsor \
            \oapp{v}{\opbox}{\xi} 
         \nsor \orapp{\opbox}{r}	
         \ifx#1\relax
            \nsor
          \else
            \nsorl
          \fi
          \olet{\gamma}{\rho}{x}{\opbox}{e} 
         \nsorl
            \oderef{\opbox} 
         \nsor
            \oassign{\opbox}{e} 
         \nsor 
            \oassign{v}{\opbox}
         \ifx#1\relax
            \nsor
         \else
            \nsorl
         \fi
            \oshare{\gamma}{\opbox} 
         \nsor 
            \orelease{\gamma}{\opbox}
         \nsorl
            \olock{\gamma}{\opbox}
         \nsor
            \ounlock{\gamma}{\opbox}
         \nsor
            \oeval{\gamma}{\opbox}
         \nsorl
         \oif{\opbox}{e_1}{e_2}
      }
}}

\newcommand{\orApp}{
   \nrule[E-A]{
      v' \aequiv \ofunc{\gamma_1}{x}{e_1}
      {  
         \tau'
      }
   }{
        \cprog;  
          T,
            \oth{n}{\theta;
               E[\oapp{v'}{v}
              {\nL{\gamma_b}}]
         }
      \oparrowt{}
      \cprog;T,\oth{n}
       {\theta;E[\oeval{\gamma_b}{e_1[v/x]}]}
   }
}

\newcommand{\orRPoly}[1][\relax]{
   \nrule[E-RP]{
      \text{fresh } n_2 
   }{
     \ifx#1\relax\else
       \begin{array}[t]{@{}l@{}}
     \fi
      \cprog;T,\oth{n}{
         \theta;E[(\opoly{\rho}{f}) [ \imath @n_1]]
      }
         \ifx#1\relax\else
            \\ \ntab
         \fi
      \oparrowt{\epsilon}
      \cprog;T,\oth{n}{\theta;E[f[\imath @ n_2/\rho]]}
     \ifx#1\relax\else
       \end{array}
     \fi
   }
}

\newcommand{\orShare}{
   \nrule[E-SH]{
      \theta(\imath) \geq (1,0)
         &
     \theta' = \fsub{\theta}{\imath}{1}{0}
   }{
      \cprog;T,\oth{n}{\theta;E[\oshare{\gamma_1}{\oloc{\imath}}]}
      \oparrowt{}
      \cprog;T,\oth{n}{{\theta'};
         E[\ounit]}
   }
}
\newcommand{\orRelease}{
   \nrule[E-RL]{
      \theta(\imath) \geq (1,0)
         &
         \theta(\imath) = (n_1,n_2)
         \\
         n_1 = 1 \Rightarrow n_2 = 0
         &
        \theta' = \fsub{\theta}{\imath}{-1}{0}
   }{
      \cprog;T,\oth{n}{\theta;E[\orelease{\gamma_1}{\oloc{\imath}}]}
      \oparrowt{}
      \cprog;T,\oth{n}{{\theta'};
         E[\ounit]}
   }
}

\newcommand{\orLockZero}{
   \nrule[E-LK0]{
         {\epsilon} = 
          \flockset{\imath}{1}
                   {E[\oeval{\gamma_1}{\opbox}]}
         &
         \theta = \theta'',\imath \mapsto n_1;0;\epsilon_1;\epsilon_2
         \\
         \theta' = \theta'',\imath \mapsto n_1;1;\mfun{dom}{S};\epsilon
         &
         n_1 \geq 1
         & \mfun{locked}{T} \cap \epsilon  = \emptyset
   }{
      \cprog;T,\oth{n}{\theta;E[\olock{\gamma_1}{\oloc{\imath}}]}
      \oparrowt{}
      \cprog;T,\oth{n}{{\theta'};
         E[\ounit]}
   }
}
\newcommand{\orLockOne}{
   \nrule[E-LK1]{
         \theta(\imath)  \geq (1,1) 
       &
          \theta' = \fsub{\theta}{\imath}{0}{1}
    }{
      \cprog;T,\oth{n}{\theta;E[\olock{\gamma_1}{\oloc{\imath}}]}
      \oparrowt{}
      \cprog;T,\oth{n}{{\theta'};
         E[\ounit]}
   }
}

\newcommand{\orUnlock}{
   \nrule[E-UL]{
         \theta(\imath) \geq (1,1)
         &
         \theta' = \fsub{\theta}{\imath}{0}{-1}
   }{
      \cprog;T,\oth{n}{\theta;E[\ounlock{\gamma_1}{\oloc{\imath}}]}
      \oparrowt{}
      \cprog;T,\oth{n}{{\theta'};
         E[\ounit]}
   }
}

\newcommand{\orNewReg}{
   \nrule[E-NG]{
      \text{fresh}\ \imath @ n_1 &
         S' = S , \imath \mapsto v				
         &
         \theta' = \theta,\imath \mapsto 1;1;\emptyset;\emptyset
   }{
      \cprog;T,\oth{n}{\theta;
        E[\olet{\gamma_2}{\rho}{x}{v}{e_2}]}
         \oparrowt{\epsilon}
         \cprog';
         T,\oth{n}{\theta';E[
            {e_2[\imath @ n_1/\rho][\oloc{\imath}/x]}]}
   }
}

\newcommand{\orEval}{
   \nrule [E-PP]{} {
      \cprog;T,\oth{n}{\theta;E[\oeval{\gamma}{v}]}
      \oparrowt{}
      \cprog;T,\oth{n}{{\theta};E[v]}
   }
}

\newcommand{\orIT}[1][\relax]{
   \nrule [E-IT]{} {
     \ifx#1\relax\else
       \begin{array}[t]{@{}l@{}}
     \fi
      \cprog;T,\oth{n}{\theta;E[\oif{\otrue}{e_1}{e_2}]}
         \ifx#1\relax\else
            \\ \ntab
         \fi
      \oparrowt{}
      \cprog;T,\oth{n}{{\theta};E[e_1]}
     \ifx#1\relax\else
       \end{array}
     \fi
   }
}

\newcommand{\orIF}[1][\relax]{
   \nrule [E-IF]{} {
     \ifx#1\relax\else
       \begin{array}[t]{@{}l@{}}
     \fi
      \cprog;T,\oth{n}{\theta;E[\oif{\ofalse}{e_1}{e_2}]}
         \ifx#1\relax\else
            \\ \ntab
         \fi
      \oparrowt{}
      \cprog;T,\oth{n}{{\theta};E[e_2]}
     \ifx#1\relax\else
       \end{array}
     \fi
   }
}

\newcommand{\orFix}[1][\relax]{
   \nrule [E-FX]{
   }{
     \ifx#1\relax\else
       \begin{array}[t]{@{}l@{}}
     \fi
      \cprog;T,\oth{n}{\theta;
            E[\oapp{\ofix{x}{\tau}{f}}{v}{\nL{\gamma_a}}]
      }
         \ifx#1\relax\else
            \\ \ntab
         \fi
      \oparrowt{}
      \cprog;T,\oth{n}{{\theta}; %
            E[\oapp{f[\ofix{x}{\tau}{f}/x]}{v}{\nL{\gamma_a}}]
         }
     \ifx#1\relax\else
       \end{array}
     \fi
   }
}

\newcommand{\orDeref}{
   \nrule[E-D]{
         \theta(\imath) \geq (1,1)
         &
         \imath \notin \mfun{locked}{T}
   }{
      \cprog;T,\oth{n}{\theta;E[\oderef{\oloc{\imath}}]}
      \oparrowt{\epsilon}
      \cprog ;T,\oth{n}{\theta;E[S(\imath)]}
   }
}

\newcommand{\orAsgn}{
   \nrule[E-AS]{
         \theta(\imath) \geq (1,1)
         &
         \imath \notin \mfun{locked}{T}
   }{
      \cprog;T,\oth{n}{\theta;E[\oassign{\oloc{\imath}}{v}]}
      \oparrowt{\epsilon}
      \cprog[\imath\mapsto v];T,\oth{n}{\theta;E[\ounit]}
   }
}

\newcommand{\listitems}[1]{
   \renewcommand{\labelitemi}{-}
   \begin{itemize}
#1
   \end{itemize}
}

\newcommand{\orSpawn}{
   \nrule[E-SN]{
      v' \aequiv 
         \ofunc{\gamma_1}{x}{e_1}
      {\tfunc{\tau_1}{\gamma_a}{\gamma_2}{\tau_2}}  
      &
         \text{fresh $n'$}	
      &
    \fsplit{\theta_1}{\theta_2}{\theta}{\fmax{\gamma_a}}
   }{
      S;T,\oth{n}{\theta;E[\oapp{v'}{v}{\nP{???}}]}
      \oparrowt{n}
         S; T,\oth{n}{{\theta_1};E[\ounit]},
         \oth{n'}{\theta_2;
            \opbox[
               \oapp{v'}{v}{\nL{\fmin{\gamma_a}}}
               ]
         }
   }
}

\newcommand{\orTerminate}{
   \nrule[E-T]{  
         \forall \imath.\, \theta(\imath) = (0,0)
   }
   { 
      \cprog;T,\oth{n}{\theta;\ounit}
      \oparrowt{n}
      \cprog; T }
}

\newcommand{\gplus}{\ensuremath{\oplus}}

\mycomment{
   \nrule[SUM0]{}{\mfun{sum}{\emptyset}}
   \rspace
      \nrule[SUM1]{
         \gamma = \gamma',\conelt{r}{\kappa;\kappa';\epsilon;\epsilon'}{}
         \\
            \gamma'' = \mfun{sum}{\gamma'} &
            r \notin \mfun{dom}{\gamma'}. 
      }{
         \mfun{sum}{\gamma'',\conelt{r}{\kappa;\kappa';\epsilon;\epsilon'}{}}
      }
   \rspace
      \nrule[SUM2]{
         &
            \forall \jjmath \in [1,4]. \neg\mfun{is\_pure}{\kappa_\jjmath}
         \\
            \gamma'' = \mfun{sum}{\gamma'} 
         &
            \gamma = \gamma',\conelt{r}{\kappa_1;\kappa_2;\epsilon_1;\epsilon_2}{},
            \conelt{r}{\kappa_3;\kappa_4;\epsilon_3;\epsilon_4}{}
      }{ 
         \mfun{sum}{\gamma'',\conelt{r}{\kappa_1+\kappa_3;\kappa_2+\kappa_4;\epsilon_1\cup\epsilon_3;\epsilon_2\cup\epsilon_4}{}}
      }
   \rline
}

\newcommand{\trVar}{
   \nrule[T-V]{
         \vdash \cstda;\gamma;\gamma
         \\
         (x:\tau') \in \Gamma
         &
         \tau \simeq \tau'
   }{
      \ttype{\cstda}{x}{\tau}{{\gamma}}{\gamma}
   }
}

\newcommand{\trUnit}{
   \nrule[T-U]{
         \vdash \cstda;\gamma;\gamma
   }{
      \ttype{\cstda}{()}{\tunit}{{\gamma}}{\gamma}
   }
}

\newcommand{\trTrue}{
   \nrule[T-TR]{
         \vdash \cstda;\gamma;\gamma
   }{
      \ttype{\cstda}{\otrue}
      {\tbool}{{\gamma}}{\gamma}
   }
}

\newcommand{\trFalse}{
   \nrule[T-FL]{
         \vdash \cstda;\gamma;\gamma
   }{
      \ttype{\cstda}{\ofalse}
      {\tbool}{{\gamma}}{\gamma}
   }
}

\newcommand{\trLoc}{
   \nrule[T-L]{
         \vdash \cstda;\gamma;\gamma
         \\
         (\imath \mapsto \tau') \in M 
         &
         \tau \simeq \tref{\tau'}{\imath}
   }{
      \ttype{\cstda}{\oloc{\imath}}
      {\tau}{{\gamma}}{\gamma}
   }
}
\newcommand{\trRPFunc}{
   \nrule[T-RF]{
      \ttype{ \cstdD{\rho} }{f}{\tau}{{\gamma}}{\gamma}
   }{
      \ttype{ \cstda}{ \opoly{\rho}{f}}
      {\tforall{\rho}{\tau}}{{\gamma}}{\gamma}
   }
}
\newcommand{\trFunc}{
   \nrule[T-F]{
      \vdash \cstda;\gamma;\gamma     
         &
         \tau' \aequiv 
         \tfunc{\tau_1}
      {\gamma_b}
      {\gamma_b}
      {\tau_2} 
      &
         M;\Delta \vdash   \tau'
         &			
         \tau \simeq \tau'
         \\
         \nL{\emptyset} \vdash \gamma_b
         \Rightarrow
         \ttype{M;\Delta;\Gamma,x:\tau_1}
               {e_1}
               {\tau_2}
               {\fmin{\gamma_b}}
               {\gamma_b}
   }{
      \ttype{\cstda}{\ofunc{}{x}{e_1}{\tau'}}
      {\tau}{\gamma}{\gamma}
   }
}

\newcommand{\trRPApp}[1][\relax]{
\ifx#1\relax%
   \nrule[T-RP]{
      \tholds{\cT}{r}  
      &
      \tholds{\cT}{\tau[\loc/\rho]}
      \\
         \ttype{\cstda}{e_1}{\tforall{\rho}{\tau}}
      {\gamma}{\gamma'}
   }{
      \ttype{\cstda}{\orapp{e_1}{\loc}}{\tau[\loc/\rho]}{\gamma}{\gamma'}
   }
\else
   \nrule[T-RP]{
      \tholds{\cT}{r}  
      &
      \tholds{\cT}{\tau[\loc/\rho]}
      &
      \ttype{\cstda}{e_1}{\tforall{\rho}{\tau}}{\gamma}{\gamma'}
   }{
      \ttype{\cstda}{\orapp{e_1}{\loc}}{\tau[\loc/\rho]}{\gamma}{\gamma'}
   }
\fi
}

\newcommand{\trAsgn}{
   \nrule[T-AS]{
      \ttype{\cstda}{e_1}{\tref{\tau}{\loc}}
      {\gamma_1}{\gamma'} 
      \\
         \ttype{\cstda}{e_2}{\tau}
      {\gamma}{\gamma_1} 
      & 
         \mfun{\gamma}{\loc} \geq (1,1)
   }{
      \ttype{\cstda}{ \oassign{e_1}{e_2}}
      {\tunit}{\gamma}{\gamma'}
   }
}

\newcommand{\trDeref}{
   \nrule[T-D]{
         \mfun{\gamma}{\loc} \geq (1,1)
      \\
      \ttype{\cstda}{e_1}{\tref{\tau}{\loc}}{\gamma}{\gamma'} 
   }{
      \ttype{\cstda}{\oderef{e_1}}{\tau}{\gamma}{\gamma'}
   }
}

\newcommand{\trNewRgn}[1][\relax]{
\ifx#1\relax
   \nrule[T-NG]{
      \ttype{\cstda}{e_1}{\tau_1}{\gamma_1 \setminus \rho}{\gamma'}
         &
         \mfun{\gamma_1}{\rho} = (1,1)
         \\
         \teholds{\cT}{\tau}{} 
         &
         \ttype{\cstdDG{\rho}{x:\tref{\tau_1}{\rho}}}{e_2}{\tau}
         { \gamma,\rho^{0,0}}
         {\gamma_1} 
   }{
      \ttype{\cstda}
      {\olet{\gamma_1}{\rho}{x}{e_1}{e_2}}{\tau}
      {\gamma}{\gamma'}
   }
\else
   \nrule[T-NG]{
      \ttype{\cstda}{e_1}{\tau_1}{\gamma_1 \setminus \rho}{\gamma'}
         &
         \mfun{\gamma_1}{\rho} = (1,1)
         &
         \teholds{\cT}{\tau}{} 
         &
         \ttype{\cstdDG{\rho}{x:\tref{\tau_1}{\rho}}}{e_2}{\tau}
         { \gamma,\rho^{0,0}}
         {\gamma_1} 
   }{
      \ttype{\cstda}
      {\olet{\gamma_1}{\rho}{x}{e_1}{e_2}}{\tau}
      {\gamma}{\gamma'}
   }
\fi
}

\newcommand{\trApp}{
   \nrule[T-A]{
      \ttype{ \cstda }{e_1}
      {\tfunc{\tau_1}{\gamma_a}{\gamma_b}{\tau_2}}
      {\gamma_3}{\gamma'} 
      &
      \xi \vdash \gamma_2 =  \gamma \gplus \gamma_a
      \\
      \ttype{ \cstda}{e_2}{\tau_1}{\gamma_2}{\gamma_3}
      &
      \xi = \nP{}  \Rightarrow \tau_2 = \tunit
   }{
      \ttype{\cstda}{\oapp{e_1}{e_2}{\xi}}{\tau_2}
      {\gamma}
      {\gamma'}
   }
}

\newcommand{\trEval}{
   \nrule[T-PP]{
      \ttype{\cstda}
      {e}
      {\tau'}
      {\fmin{\gamma_b}}
      {\gamma_b}
      &
      \gamma_b \simeq \gamma_b'
      \\
      \nL{\gamma} \vdash \gamma' =  \gamma \gplus \gamma_b'
      &
      \tau' \simeq \tau
      &
      \vdash M;\Delta;\Gamma;\gamma;\gamma'
   }{
      \ttype{\cstda}{\oeval{\gamma}{e}}
      {\tau}{\gamma}{\gamma'}
   }
}

\newcommand{\trShare}{
   \nrule[T-SH]{
      \ttype{\cstda}
      {e}
      {\tref{\tau}{\loc}}
      {\gamma,
        r^{\kappa - (1,0)}
      }
      {\gamma'}
      \\ 
         \kappa \geq (2,0)
         &
         \mfun{\gamma}{\loc} = \kappa
   }{
      \ttype{\cstda}{\oshare{\gamma}{e}}
      {\tunit}{\gamma}{\gamma'}
   }
}

\newcommand{\trRelease}{
   \nrule[T-RL]{
      \ttype{\cstda}
      {e}
      {\tref{\tau}{\loc}}
      {\gamma,
      r^{\kappa + (1,0)}
      }
      {\gamma'}
      \\
      \kappa = (n_1,n_2) 
      &
       n_1 = 0 \Rightarrow n_2 = 0
       &
         \mfun{\gamma}{\loc} = \kappa
   }{
      \ttype{\cstda}{\orelease{\gamma}{e}}
      {\tunit}{\gamma}{\gamma'}
   }
}

\newcommand{\trLock}{
   \nrule[T-LK]{
      \ttype{\cstda}
      {e}
      {\tref{\tau}{\loc}}
      {\gamma,
            r^{\kappa- (0,1)}}
      {\gamma'}
      \\
         \kappa \geq (1,1)
         &
         \mfun{\gamma}{\loc} = \kappa
   }{
      \ttype{\cstda}{\olock{\gamma}{e}}
      {\tunit}{\gamma}{\gamma'}
   }
}

\newcommand{\trUnlock}{
   \nrule[T-UL]{
      \ttype{\cstda}
      {e}
      {\tref{\tau}{\loc}}
      {\gamma,
         r^{\kappa+(0,1)}
      }
      {\gamma'}
      \\
       \kappa \geq (1,0)
       &
         \mfun{\gamma}{\loc} = \kappa
   }{
      \ttype{\cstda}{\ounlock{\gamma}{e}}
      {\tunit}{\gamma}{\gamma'}
   }
}

\newcommand{\trThread}{
      \nrule[EA]{
         \cstda \vdash e : 
            \tau \teff{\gamma_a}{\gamma_b}
         &
            \cstda \vdash E : 
            \tefunc{\tau}{\gamma_a}{\gamma_b}{\tunit} 
         \teff{\gamma_1}{\gamma_2}
         \\
         &
         \forall r^\kappa \in \gamma_1.\, \kappa = (0,0) 
         &
         \mfun{counts\_ok}{E[\oeval{\gamma_b}{\opbox}],\theta}
         &
         \mfun{lockset\_ok}{E[\oeval{\gamma_b}{\opbox}],\theta}
      }{
         \cstda \vvdash \theta; E[e] :
            \tunit 
            \teff{\gamma_1}{\gamma_2}
      }
}

\newcommand{\trThreads}{
   \nrule{ 
   } {
      \tholds{M}{\emptyset}
   }
   \rspace
      \nrule{
         M;\emptyset;\emptyset \vvdash \theta;e : \tunit 
            \teff{\gamma}{\gamma'}
         &
            \tholds{M}{T} 
         &
            n \notin  \mfun{dom}{T}
      } 
   { 
      M \vdash {T,\oth{n}{\theta;e}}
   }
}

\newcommand{\trConfig}{

   \nrule{
      \tholds{M}{T} 
      &
         M \vdash S
      &
         \mfun{mutex}{T}
   }{
      M \vdash S;T
   }
}

\newcommand{\trStore}{
   \nrule{
      \mfun{dom}{M} = \mfun{dom}{S} 
      &
         \forall (\imath \mapsto \tau) \in M.
         \ttype{M;\emptyset;\emptyset}{S(\imath)}{\tau}{\emptyset}{\emptyset}
   }{
      M \vdash S
   }
}

   \ifllncs
   \else
      \newenvironment{definition}[1][Definition]{\begin{trivlist}
         \item[\hskip \labelsep {\bfseries #1}]}{\end{trivlist}}
    \fi

    \newif\iftechrep

\renewcommand{\labelitemi}{Case}

\newlength{\ncasewidth}
\newcommand{\casesplit}[1]%
{\bgroup%
  \begin{list}{Case}{%
    \settowidth\ncasewidth{Case}%
    \setlength{\labelwidth}{\ncasewidth}%
    \setlength{\leftmargin}{\ncasewidth}%
    \settowidth\ncasewidth{\ }%
    \setlength{\labelsep}{\ncasewidth}%
    \addtolength{\leftmargin}{\ncasewidth}%
  }
  #1
  \end{list}%
  \egroup%
}

\newcommand{\sfig}[2]{
  \subsection{#2}
  \bgroup\footnotesize%
   #1
  \egroup
}

\newcommand{\trFormalismTable}[3]{
  \subsection{#2}
  \def\trFormalismItem##1##2{%
    \item%
      \begin{tabular}[t]{@{}p{#1}@{}}%
        ##1%
      \end{tabular}%
      ##2
  }
  \begin{list}{}{%
    \setlength{\labelwidth}{0cm}%
    \setlength{\labelsep}{0cm}%
    \setlength{\itemindent}{-#1}%
    \setlength{\itemsep}{3pt plus 1pt minus 1pt}%
    \setlength{\parsep}{0pt}
    \setlength{\leftmargin}{#1}%
  }
       #3
  \end{list}
}

\newcommand{\trTable}{
  \trFormalismTable{5cm}{Formalism Summary: Operational Semantics}
  {\small
    \trFormalismItem{$\mfun{locked}{T}$}{%
      takes a list of threads $T$ and
      returns a set of locations
      locked by threads in $T$.
    }

    \trFormalismItem{$\fsub{\theta}{\imath}{n_1}{n_2}$}{%
      updates the map $\theta$ so that
      the reference and lock counts of $\theta(\imath)$
      are incremented by $n_1$ and $n_2$ respectively.
    }

    \trFormalismItem{$\mfun{\theta}{\imath}$}{%
      returns the reference and lock counts of $\theta(\imath)$.
    }

\ifeptcs\else
    \trFormalismItem{$\mfun{lk}{\kappa}$}{%
      takes a capability $\kappa$ and returns
      the lock count of $\kappa$. 
    }
    \trFormalismItem{$\mfun{locked}{\gamma}$}{%
      takes an ordered list of events $\gamma$
      and returns a set of locations
      $\epsilon$, whose lock count is 
      incremented (at some point) in $\gamma$. 
    }
\fi

    \trFormalismItem{
      $\fsplit{\theta_1}{\theta_2}{\theta}{\fmax{\gamma_a}}
      $}{%
      takes $\gamma_a$ (the effect of a new thread) and $\theta$
      and returns $\theta_1$ and $\theta_2$, such that
      the sum of the counts of each location in $\theta_1$ and
      $\theta_2$ equals the counts of the same location in $\theta$.
    }

\ifeptcs\else
    \trFormalismItem{$(\epsilon,n')= \fspan{\imath}{n}{\gamma}$}{%
      the lockset computation function takes three
      arguments, the location $\imath$, the expected
      number of unlock operations $n$ and the frame
      annotation $\gamma$ and returns a lockset 
      $\epsilon$ along with the unlock operations that
      were not found in $\gamma$. The lockset $\epsilon$
      contains locations that were locked while $n$ is
      positive.
    }
\fi

    \trFormalismItem{$\flockset{\imath}{n}{E}$}{%
      traverses the evaluation context $E$ and returns 
      the future lockset for $\imath$ acquired $n$ times,
      only examining frames of the form $\oeval{\gamma}{\opbox}$.
      The traversal ends when $E$ is empty or $n$ is zero.
    }

\ifeptcs\else
    \trFormalismItem{$S;T \oparrowt{} S;T'$}{%
      the reduction relation defines the operational
      semantics of multi-threaded programs by re-writing 
      terms and the store $S$.
    }
    \trFormalismItem{$e[v/x]$}{%
      the standard term-level substitution function
      maps expessions to expressions by substituting
      occurences of $x$ with $v$.
    }
    \trFormalismItem{$\tau[r/\rho]$}{%
      the standard type-level substitution function
      maps types to types by substituting
      occurences of $r$ with $\rho$.
    }
\fi
  }

  \trFormalismTable{5cm}{Formalism Summary: Static Semantics}
  {\small
    \trFormalismItem{$M;\Delta \vdash \tau$}{%
      well-formedness judgement
      within a typing context $M;\Delta$ for type $\tau$.
    }
\ifeptcs\else
    \trFormalismItem{$M;\Delta \vdash \gamma$}{%
      well-formedness judgement
      within a typing context $M;\Delta$ for effect $\gamma$.
    }
\fi
    \trFormalismItem{$M;\Delta \vdash r$}{%
      well-formedness judgement
      within a typing context $M;\Delta$ for location $r$.
    }
\ifeptcs\else
    \trFormalismItem{$M;\Delta \vdash \Gamma$}{%
      well-formedness judgement
      within a typing context $M;\Delta$ for term variable context $\Gamma$.
    }
    \trFormalismItem{$\vdash M$}{%
      well-formedness judgement
      for constant location typing context $M$.
    }
\fi
    \trFormalismItem{$\vdash M;\Delta;\Gamma;\gamma_1;\gamma_2$}{%
      well-formedness judgement
      for typing context $M;\Delta;\Gamma$ and effect
      $(\gamma_1;\gamma_2)$.
    }
    \trFormalismItem{$\xi \vdash \gamma$}{%
      ensures that
      pure capabilities are not aliased within $\gamma$.
      In the case of parallel application 
      (i.e.,\, $\xi=\nP{}$), the ending capability
      of each location must be zero,
      whereas the starting capability of each location
      must have a zero lock count
      when that capability is impure. 
    }
    \trFormalismItem{$\gamma(r)$}{%
      returns the most recent (i.e., rightmost)
      occurence of $r$ within effect $\gamma$.
    }
\ifeptcs\else
    \trFormalismItem{$\mfun{is\_pure}{\kappa}$}{%
      true when $\kappa$ is of the form ${n_1,n_2}$.
    }
    \trFormalismItem{$\mfun{set}{\gamma}$}{%
      true when there exist no
      duplicate locations in $\gamma$ or branch
      effects (i.e.,\, $\gamma_1 \tfsep \gamma_2$).
    }
\fi
    \trFormalismItem{$\mfun{max}{\gamma'}$}{%
      returns a subset of $\gamma'$, 
      say $\gamma$ such that
      no duplicate locations or branches exist,
      the domain of $\gamma'$ 
       equals the domain of $\gamma$ and
      each element of $\gamma$ is equal to $\gamma(r)$
      for any $r$ in the domain of $\gamma$. 
    }
    \trFormalismItem{$\fmin{\gamma'}$}{%
      takes $\gamma'$ and
      returns a \emph{prefix} $\gamma'$ of $\gamma$ 
      such that
      no duplicate locations or branches exist and
      the domain of $\gamma'$ equals the domain of $\gamma$.
    }
    \trFormalismItem{$\gamma \setminus r$}{%
      takes $\gamma$ and 
      $r$ and removes all occurences of 
      $r'$ from $\gamma$
      such that $r'$ is identical to $r$ modulo
      the tags of constant locations.
    }
    \trFormalismItem{$\xi \vdash 
                      \gamma' = \gamma \gplus \gamma_1$}{%
      takes $\gamma$, representing the environment effect \emph{before}
      a function call,
      the function effect $\gamma_1$ and 
      yields the environment effect 
      $\gamma'$ 
      representing the environment effect \emph{after}
      the function call.
      $\gamma$ is a prefix of $\gamma'$ and the suffix
      of $\gamma'$ is an adjusted version of 
      $\gamma_1$: the order of locations is the same as
      in $\gamma_1$ but the counts may be greater 
      than the ones in $\gamma_1$ as some counts may
      have been abstracted withing the scope of
      the function. It also enforces $\xi \vdash \gamma$.
      }

\ifeptcs
    \trFormalismItem{$\kappa \geq \kappa'$}{%
      true if both counts of $\kappa$ are no
      smaller than the correspoding counts of $\kappa'$.
    }
    \trFormalismItem{$\kappa + \kappa'$, $\kappa - \kappa'$}{%
      calculate the sum and difference of two
      capabilities (considered here as two-dimensional vectors).
    }
    \trFormalismItem{$\mfun{summary}{\gamma}$}{%
      used primarily for calculating the summarized effects of
      recursive functions.
    }
\fi

\ifeptcs\else
    \trFormalismItem{$r \simeq r'$}{%
      true when the locations
      resulting after removing tags from $r$ and $r'$ 
      are equivalent.
    }

    \trFormalismItem{$\gamma \simeq \gamma'$}{%
      true when $\gamma$ and $\gamma'$
      are structurally equivalent up to $\simeq$ for locations.
    }
\fi
    \trFormalismItem{$\tau \simeq \tau'$}{%
      true when $\tau$ and $\tau'$
      are structurally equivalent after removing $@n$ annotations
      from locations.
    }
\ifeptcs\else
    \trFormalismItem{$M;\Delta;\Gamma \vdash e:\tau\teff{\gamma}{\gamma'}$}{%
      the expression typing judgement that takes
      the typing context $M;\Delta;\Gamma$, the 
      expression $e$ and the input effect $\gamma$ and 
      returns the type $\tau$ and output effect $\gamma'$
      assigned to $e$.
    }
\fi
    \trFormalismItem{$M;\Delta;\Gamma \vdash E : 
      \tefunc{\tau}{\gamma_a}{\gamma_b}{\tau'}
      \teff{\gamma_1}{\gamma_2}
      $}{%
      the evaluation typing context judgement that 
      takes the typing context 
      $M;\Delta;\Gamma$, the evaluation context $E$,
      the expected effect $(\gamma_a;\gamma_b)$
      and the expected type $\tau$ 
      (for the innermost hole in $E$),
      the input effect $\gamma_1$ and returns the type
      $\tau'$ and the effect $\gamma_2$
      that will be returned by $E$ when it is 
      filled with an expression of type $\tau$ and effect
      $(\gamma_a;\gamma_b)$.
    }
  }

  \trFormalismTable{5cm}{Formalism Summary: Type Safety}
  {\small
    \trFormalismItem{$\mfun{blocked}{T,n}$}{%
      true when thread $n$ of thread
      list $T$ is in a blocked (i.e,\, waiting for a lock)
      state.
    }
    \trFormalismItem{$\mfun{mutex}{T}$}{%
      true when each lock is held
      by at most one thread of $T$.
    }
\ifeptcs\else
    \trFormalismItem{$\mfun{deadlocked}{T}$}{%
      true when there exists a subset
      of $T$, say $T'$, such that each thread in $T'$
      has acquired a lock and is waiting for a lock that
      is acquired by another thread in $T'$.
    }
    \trFormalismItem{$M;\Delta;\Gamma \vdash F : 
      \tefunc{\tau}{\gamma_a}{\gamma_b}{\tau'}
      \teff{\gamma_1}{\gamma_2}
      $}{%
      similar to the previous judgement. 
    }
    \trFormalismItem{$M;\Delta;\Gamma \vdash
      E[e] : \tau \teff{\gamma}{\gamma'}$}{%
      extends the typing judgement with evaluation contexts.
    }
    \trFormalismItem{$M;\Delta;\Gamma \vvdash
      \theta;E[e] : \tau \teff{\gamma}{\gamma'}$}{%
      this judgement types the evaluation
      context $E$, the expression $e$ 
      and establishes a correspondence between
      the \emph{access list} $\theta$ and the static effect $\gamma'$.
    }
\fi
\mycomment{
    \trFormalismItem{$\theta;\gamma \vdash \theta'$}{%
      this judgement takes the access list $\theta$ and
      an effect $\gamma$ and subtracts the capability
      for each location in $\gamma$ from the corresponding
      location of $\theta$. The resulting access list is 
      $\theta'$. 
    }
   \trFormalismItem{$E;\gamma \vdash_{st} \theta;\epsilon$}{%
      this judgement takes an evaluation context $E$ and
      the effect of the outermost
      context enclosed by $E$ ($\gamma$)
      and outputs a $\theta$ and $\epsilon$ such that
      the counts of each location in $\theta$ match the 
      the sum of counts of all \emph{residual effects} of $E$
      and the locations in $\epsilon$ represent locations
      having pure capabilities with positive counts.
    }

   \trFormalismItem{$\gamma\vdash_{nl} \theta;\epsilon'$}{%
      this judgement takes $\gamma$ and $\epsilon$ and
      outputs a $\theta$ and $\epsilon'$ such that
      the counts of $\gamma$ equal the counts of $\theta$
      and all locations in $\epsilon$ correspond to 
      pure capabilities with positive counts. 
  }
}

    \trFormalismItem{$\mfun{counts\_ok}{E,\theta}$}{%
      takes an evaluation context $E$
      and an access list $\theta$ and holds
      when the sum of all $\cfont{pop}$ expression annotations
      in $E$
      equal the counts of $\theta$. 
      It establishes an exact correspondence
      between dynamic and static counts.
    }

    \trFormalismItem{$\mfun{lockset\_ok}{E,\theta}$}{%
      takes an evaluation context $E$
      and an access list $\theta$ and holds
      when
      the future lockset ($\mathsf{lockset}$ function) 
      of an acquired lock at any program point
      is \emph{always} a subset of the future lockset
      computed when the lock was initially acquired.
    }
\ifeptcs\else
    \trFormalismItem{$M \vdash S $}{%
      store typing judgement. Each value assigned to 
      location $\imath$ must be well typed in the 
      context $M;\emptyset;\emptyset$ with type $M(\imath)$.
    }
    \trFormalismItem{$M \vdash T,n:\theta;E[e]$}{%
      the thread typing judgement is true,
      when $M;\emptyset;\emptyset 
      \vvdash \theta;E[e] :\tau \teff{\gamma}{\gamma'}$
      holds for each thread. 
    }
    \trFormalismItem{$M \vdash S;T$}{%
      the configuration typing judgement is true when
      the threads of $T$ are well-typed, the store $S$
      is well-typed, and $\mfun{mutex}{T}$ holds 
      (i.e.,\, locks are held by at most one thread of $T$).
    }
    \trFormalismItem{$ \vdash S;T$}{%
      this judgement is true when $S;T$ is 
      \emph{not stuck}.
      That is, each thread can take a step or is waiting
      for some lock and there exists no subset of $T'$ such
      that threads in that subset are deadlocked.
    }
    \trFormalismItem{$ S;T \oparrowt{n}^n S';T'$}{%
      this relation extends the reduction relation so
      that $S;T$ is transformed to $S';T'$ in $n$ steps.
    }
\fi
  }
}

\renewcommand{\rline}[1][1.25em]{\ensuremath{\\[#1]}\footnotesize}
\renewcommand{\rspace}{\ensuremath{\hspace{1.5em}}}

\newcommand\Figref[1]{Figure~\ref{#1}}
\newcommand\figref[1]{Figure~\ref{#1}}

\newcommand{\ignore}[1]{}

\newcommand{\cf}[1]{\leavevmode\raise.2ex\hbox{$\scriptscriptstyle\ll$}#1\,\leavevmode\raise.2ex\hbox{$\scriptscriptstyle\gg$}}

\author{%
  Prodromos Gerakios 
  \hspace*{1em}
  Nikolaos Papaspyrou
  \hspace*{1em}
  Konstantinos Sagonas 
\institute{%
  School of Electrical and Computer Engineering,
  National Technical University of Athens, Greece
}
\email{$\{\,$pgerakios$,\,$nickie$,\,$kostis$\,\}\,$@softlab.ntua.gr}
}
\def\authorrunning{P. Gerakios, N. Papaspyrou, and K. Sagonas}

\renewcommand{\labelitemi}{-}
\begin{document}
\let\orig@Itemize =\itemize         
\let\orig@Enumerate =\enumerate
\let\orig@Description =\description
\def\Nospacing{\itemsep=0pt\topsep=0pt\partopsep=0pt%
\parskip=0pt\parsep=0pt}

\def\noitemsep{
\renewenvironment{itemize}{\orig@Itemize\Nospacing}{\endlist}
\renewenvironment{enumerate}{\orig@Enumerate\Nospacing}{\endlist}
\renewenvironment{description}{\orig@Description\Nospacing}%
{\endlist}
}

\def\doitemsep{
\renewenvironment{itemize}{\orig@Itemize}{\endlist}
\renewenvironment{enumerate}{\orig@Enumerate}{\endlist}
\renewenvironment{description}{\orig@Description}{\endlist}
}

\maketitle

\newcommand\myparagraph[1]{%
  \par\vskip 3pt plus 3pt%
  \noindent\textbf{#1.}\hskip 1em plus 0.25em minus 0.25em%
}

\begin{abstract}
  Deadlocks occur in concurrent programs as a consequence of
  cyclic resource acquisition between threads. In this paper we
  present a novel type system that guarantees deadlock freedom for
  a language with references, unstructured locking primitives, and
  locks which are implicitly associated with references.
  The proposed type system does not impose a strict lock acquisition
  order and thus increases programming language expressiveness.
\end{abstract}

\section{Introduction}

Lock-based synchronization may give rise to deadlocks. Two or more
threads are deadlocked when each of them is waiting for a lock that
is acquired by another thread.
According to Coffman \textit{et al.}~\cite{SysDeadlocks@CompSurv-71}, a set of
threads reaches a \emph{deadlocked state}
when the following conditions hold:
\bgroup\noitemsep
\begin{itemize}
  \item \emph{Mutual exclusion}:
    Threads claim exclusive control of the locks that they acquire.
  \item \emph{Hold and wait}:
    Threads already holding locks may request (and wait for) new locks.
  \item \emph{No preemption}:
    Locks cannot be forcibly removed from threads; they must be released
    explicitly by the thread that acquired them.
  \item \emph{Circular wait}:
    Two or more threads form a circular chain, where each thread waits
    for a lock held by the next thread in the chain.
\end{itemize}
\egroup
Coffman has identified three strategies that guarantee deadlock-freedom
by denying at least one of the above conditions \emph{before} or
\emph{during} program execution:
\bgroup\noitemsep
\begin{itemize}
 \item \emph{Deadlock prevention}:
   At each point of execution, \emph{ensure} that at least one of the
   above conditions is not satisfied. Thus, programs that fall into this
   category are correct by design.
 \item \emph{Deadlock detection and recovery}:
   A dedicated observer thread \emph{determines} whether
   the above conditions are satisfied and preempts some of
   the deadlocked threads, releasing (some of) their locks,
   so that the remaining threads can make progress.
 \item \emph{Deadlock avoidance}:
   Using information  that is computed in advance regarding thread
   resource allocation, \emph{determine} whether granting a lock will
   bring the program to an \emph{unsafe} state, i.e.,\ a state which
   can result in deadlock, and only grant locks that lead to safe states.
\end{itemize}
\egroup

Several type systems have been proposed that guarantee deadlock freedom,
the majority of which is based on the first two strategies.
In the deadlock prevention category,
one finds type and effect systems
that guarantee deadlock freedom
by statically enforcing a global lock acquisition order
that must be respected by all threads~\cite{%
  FlanaganAbadi@ESOP-99,%
  OwnershipTypes@OOPSLA-02,%
  Kobayashi@CONCUR-06,%
  NonLexicalDeadlock@APLAS-08,%
  Vasco@PLACES-09%
}.
In this setting,
lock handles are associated with type-level lock names via the use of
singleton types.
Thus, handle $lk_\imath$ is of type $\cfont{lk}(\imath)$. The same applies
to lock handle variables.
The effect system tracks the order of lock operations on handles or
variables and determines whether all threads acquire locks in the same
order.

Using a strict lock acquisition order is a constraint we want to
avoid, as it unnecessarily rejects many correct programs.
It is not hard to come up with an example that shows that
imposing a partial order on locks is too restrictive. The simplest of
such examples can be reduced to program fragments of the form:
\begin{ndisplay}
  (\cfont{lock} \ x \ \cfont{in} \ \ldots \
   \cfont{lock} \ y \ \cfont{in} \ \ldots)
  \ \ || \ \
  (\cfont{lock} \ y \ \cfont{in} \ \ldots \
   \cfont{lock} \ x \ \cfont{in} \ \ldots)
\end{ndisplay}
In a few words, there are two parallel threads which acquire two different
locks, $x$ and $y$, in reverse order.
When trying to find a partial order $\le$ on locks for this program,
the type system or static analysis tool
will deduce that $x \le y$ must be true, because of the first thread,
and that $y \le x$ must be true, because of the second.
Thus, the program will be rejected, both in the system of Flanagan and
Abadi which requires annotations~\cite{FlanaganAbadi@CONCUR-99} and in
the system of Kobayashi which employs
inference~\cite{Kobayashi@CONCUR-06} as there is no single lock order
for \emph{both} threads.
Similar considerations apply to the more recent works of
Suenaga~\cite{NonLexicalDeadlock@APLAS-08} and Vasconcelos
\textit{et al.}~\cite{Vasco@PLACES-09} dealing with non
lexically-scoped locks.

Our work follows the third strategy (deadlock avoidance).
%
It is based on an idea put forward recently by Boudol,
who proposed a type system for deadlock avoidance that
is more permissive than existing approaches~\cite{Boudol@ICTAC-09}.
However, his system is suitable for programs that use
\emph{exclusively} lexically-scoped locking primitives.
In this paper we present a simple language with functions, mutable
references, explicit (de-)allocation constructs and unstructured
(i.e.,\ non lexically-scoped) locking primitives.
%
Our approach ensures deadlock freedom for the proposed language by 
preserving exact information about the order of events,
both statically and dynamically.
It also forms the basis for a much simpler approach to providing
deadlock freedom, following a quite different path, that is
easier to program and amenable to type inference, which
has been implemented for C/pthreads~\cite{DeadlockAvoidance@TLDI-11}.

In the next section, we informally describe Boudol's idea and
present an informal overview of our type and effect system.
In Section~\ref{sec:formalism} we formally define the syntax
of our language, its operational semantics and the type and
effect system. In Section~\ref{sec:safety} we reason
about the soundness of our system and the paper ends with a few
concluding remarks.

\section{Deadlock Avoidance} \label{sec:avoidance}
Recently, Boudol developed a type and effect system for deadlock
freedom~\cite{Boudol@ICTAC-09}, which is based on \emph{deadlock
  avoidance}.
The effect system calculates
for each expression the set of acquired locks
and annotates lock operations with the ``future'' lockset.
The runtime system utilizes the inserted
annotations so that each lock operation can only
proceed when its ``future'' lockset is unlocked.
The main advantage of Boudol's type system is that it allows a
larger class of programs to type check and thus
increases the programming language expressiveness
as well as
concurrency by allowing arbitrary locking schemes.

The previous example can be rewritten in Boudol's language as follows,
assuming that the only lock operations in the two threads are those
visible:
\begin{ndisplay}
  (\cfont{lock}_{\{y\}} \ x \ \cfont{in} \ \ldots \
   \cfont{lock}_{\emptyset} \ y \ \cfont{in} \ \ldots)
  \ \ || \ \
  (\cfont{lock}_{\{x\}} \ y \ \cfont{in} \ \ldots \
   \cfont{lock}_{\emptyset} \ x \ \cfont{in} \ \ldots)
\end{ndisplay}
This program is accepted by Boudol's type system which, in general,
allows locks to be acquired in \emph{any} order.
At runtime, the first lock operation of the first thread must ensure
that $y$ has not been acquired by the second (or any other) thread,
before granting $x$ (and symmetrically for the second thread).
The second lock operations need not ensure anything special, as the
future locksets are empty.

The main disadvantage of Boudol's work is that locking operations have
to be lexically-scoped. In fact, as we show here, even if Boudol's
language had \cfont{lock}/\cfont{unlock} constructs, instead
of \cfont{lock}$\,\ldots\,$\cfont{in}$\,\ldots$, his type system is
not sufficient to guarantee deadlock freedom.
The example program in \figref{fig:example:boudol}(a) will help us see why:
It updates the values of three shared variables, $x$, $y$ and $z$,
making sure at each step that only the strictly necessary locks are held.%
\footnote{To simplify presentation, we assume here that there is one
  implicit lock per variable, which has the same name.
  This is more or less consistent with our formalization in
  Section~\ref{sec:formalism}.}
\newcommand{\dista}{6cm}
\newcommand{\fstmpdist}{.5\textwidth}
\newcommand{\sndmpdist}{.3\textwidth}
\begin{figure}[t]
  \small
  \centering
  \begin{minipage}[t]{\fstmpdist}
  \begin{ndisplay}[noindent]
    \cfont{let} \nabstab{1cm}
    f \ \cfont{=} \
    \lambda \, x. \,
    \lambda \, y. \,
    \lambda \, z. \ntab
    \nbox{%
      \cfont{lock}_{\{y\}} \ x; \nabstab{\dista}
      x := x+1; \\
      \cfont{lock}_{\{z\}} \ y; \nabstab{\dista}
      y := y+x; \\
      \cfont{unlock}\ x; \\
      \cfont{lock}_{\emptyset}\ z; \nabstab{\dista}
      z := z+y; \\
      \cfont{unlock}\ z; \\
      \cfont{unlock}\ y
    } \\
    \cfont{in} \nabstab{1cm}
    f\ a\ a\ b
  \end{ndisplay}
  \vskip -6pt
  \centering (a) before substitution
  \end{minipage}
  \vline
  \begin{minipage}[t]{\sndmpdist}
  \begin{ndisplay}
    \cfont{lock}_{\{a\}} \ a; \nabstab{2cm}
    a := a+1; \\
    \cfont{lock}_{\{b\}} \ a; \nabstab{2cm}
    a := a+a; \\
    \cfont{unlock}\ a; \\
    \cfont{lock}_{\emptyset}\ b; \nabstab{2cm}
    b := b+a; \\
    \cfont{unlock}\ b; \\
    \cfont{unlock}\ a \\
  \end{ndisplay}
  \vskip -6pt
  \centering (b) after substitution
  \end{minipage}
  \caption{An example program,
    which is well typed before substitution (a)
    but not after (b).\label{fig:example:boudol}}
\end{figure}%

In our na\"{i}vely extended (and broken, as will be shown)
version of Boudol's system,
the program in \figref{fig:example:boudol}(a) will type check.
The future lockset annotations of the three locking
operations in the body of $f$ are $\{y\}$, $\{z\}$ and $\emptyset$,
respectively.
(This is easily verified by observing the lock operations
between a specific \cfont{lock}/\cfont{unlock} pair.)
Now, function $f$ is used by instantiating both $x$ and $y$ with the
same variable $a$, and instantiating $z$ with a different variable $b$.
The result of this substitution is shown in
\figref{fig:example:boudol}(b).
The first thing to notice is that, if we want this program to work
in this case, locks have to be \emph{re-entrant}.
This roughly means that if a thread holds some lock,
it can try to acquire the same lock again;
this will immediately succeed, but then the thread will
have to release the lock \emph{twice}, before it is
actually released.

Even with re-entrant locks, however, the
program in \figref{fig:example:boudol}(b)
does not type check with the present annotations.
The first \cfont{lock} for $a$ now matches with the \emph{last}
(and not the first) \cfont{unlock};
this means that $a$ will remain locked during the whole execution
of the program.
In the meantime $b$ is locked, so the future lockset annotation
of the first \cfont{lock} should contain $b$, but it does not.
(The annotation of the second \cfont{lock} contains $b$,
but blocking there if lock $b$ is not available does not prevent
a possible deadlock; lock $a$ has already been acquired.)
So, the technical failure of our na\"{i}vely extended language
is that the preservation lemma breaks.
From a more pragmatic point of view, if a thread running in
parallel already holds $b$ and, before releasing it, is about
to acquire $a$, a deadlock can occur.
The na\"ive extension of Boudol's system also fails for another
reason: it is based on the assumption that calling a function
cannot affect the set of locks held by a thread.
This is obviously not true, if non lexically-scoped locking
is to be supported.

The type and effect system proposed in this paper
supports unstructured locking, by preserving
more information at the effect level.
%
Instead of treating effects as unordered collections of locks,
our type system precisely tracks effects as an
order of \cfont{lock} and \cfont{unlock}
operations, without enforcing a strict lock-acquisition order.
The \emph{continuation effect} of a term represents the effect of the
function code succeeding that term.
In our approach, lock operations
are annotated with a continuation effect.
When a \cfont{lock} operation is evaluated, the future lockset
is calculated by inspecting its continuation effect.
The \cfont{lock} operation succeeds only when
both the lock and the future lockset are available.

\Figref{fig:example:ours} illustrates the same program as
in \figref{fig:example:boudol},
except that locking operations are now annotated
with continuation effects.
%
\newcommand{\distb}{7cm}%
\newcommand{\distc}{3.5cm}%
\newcommand{\distlam}{1cm}%
\newcommand{\distlamb}{3.5cm}%
\newcommand{\sndmpdistba}{.55\textwidth}%
\newcommand{\sndmpdistbb}{.35\textwidth}%
\newcommand\then{\ensuremath{,\linebreak[0]\,}}%
\begin{figure}[t]
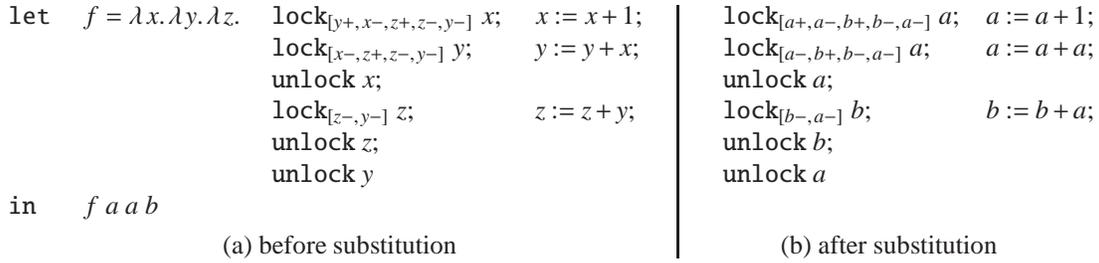

  \small
  \centering
  \begin{minipage}[t]{\sndmpdistba}
  \begin{ndisplay}[noindent]
    \cfont{let} \nabstab{\distlam}
    f \ \cfont{=} \
    \lambda \, x. \,
    \lambda \, y. \,
    \lambda \, z.
    \nabstab{\distlamb}
    \nbox{%
      \cfont{lock}_{[y+\then x- \then z+ \then z- \then y-]} \ x; \nabstab{\distb}
      x := x+1; \\
      \cfont{lock}_{[x- \then z+ \then z- \then y-]} \ y; \nabstab{\distb}
      y := y+x; \\
      \cfont{unlock}\ x; \\
      \cfont{lock}_{[z- \then y-]} \ z; \nabstab{\distb}
      z := z+y; \\
      \cfont{unlock}\ z; \\
      \cfont{unlock}\ y
    } \\
    \cfont{in} \nabstab{1cm}
    f\ a\ a\ b
  \end{ndisplay}
  \vskip -6pt
  \centering (a) before substitution
  \end{minipage}
  \vline
  \begin{minipage}[t]{\sndmpdistbb}
  \begin{ndisplay}
    \cfont{lock}_{[a+\then a- \then b+ \then b- \then a-]} \ a; \nabstab{\distc}
    a := a+1; \\
    \cfont{lock}_{[a- \then b+ \then b- \then a-]} \ a; \nabstab{\distc}
    a := a+a; \\
    \cfont{unlock}\ a; \\
    \cfont{lock}_{[b- \then a-]}\ b; \nabstab{\distc}
    b := b+a; \\
    \cfont{unlock}\ b; \\
    \cfont{unlock}\ a \\
  \end{ndisplay}
  \vskip -6pt
  \centering (b) after substitution
  \end{minipage}
  \caption{The program of \figref{fig:example:boudol}
    with continuation effect annotations;
    now well typed in both cases.\label{fig:example:ours}}
\end{figure}%
For example, the annotation $[y+\then x-\then z+ \then z- \then y-]$
at the first \cfont{lock}
operation means that in the future (i.e.,\ after this lock operation)
$y$ will be acquired, then $x$ will be released, and so on.%
\footnote{In the examples of this section, a simplified version of
  effects is used, to make presentation easier.
  In the formalism of Section~\ref{sec:formalism}, the plus and minus
  signs would be encoded as differences in lock counts, e.g.,
  $y+$ would be encoded by a $y^{1,0}$ (an unlocked $y$)
  followed in time by a $y^{1,1}$ (a locked $y$).}
If $x$ and $y$ were different, the runtime system would deduce
that between this \cfont{lock} operation on $x$ and the
corresponding \cfont{unlock} operation, only $y$ is locked,
so the future lockset in Boudol's sense would be $\{y\}$.
On the other hand, if $x$ and $y$ are instantiated with the same
$a$, the annotation becomes $[a+ \then a- \then b+ \then b- \then a-]$
and the future lockset that is calculated is now the correct $\{a, b\}$.
In a real implementation, there are several optimizations that can be
performed (e.g.,\ pre-calculation of effects) but we do not deal with
them in this paper.

There are three issues that must be faced, before we can apply this
approach to a full programming language.
First, we need to consider continuation effects in an interprocedural
manner: it is possible that a lock operation in the body of function
$f$ matches with an unlock operation in the body of function $g$
after the point where $f$ was called, directly or indirectly.
In this case, the future lockset
for the lock operation may contain locks that are not visible in the
body of $f$.
We choose to compute function effects intraprocedurally and to
annotate each application term with a continuation effect,
which represents the effect of the code succeeding the application
term in the calling function's body.
A runtime mechanism pushes information about continuation effects
on the stack and, if necessary, uses this information to correctly
calculate future locksets, taking into account the continuation
effects of the enclosing contexts.

Second, we need to support conditional statements.
The tricky part here is that, even in a simple conditional
statement such as
\begin{ndisplay}
  \cfont{if}\ c
  \ \cfont{then}\ (\cfont{lock}\ x; ...\ \cfont{unlock}\ x)
  \ \cfont{else}\ (\cfont{lock}\ y; ...\ \cfont{unlock}\ y)
\end{ndisplay}
the two branches have different effects: $[x+\then x-]$ and
$[y+\then y-]$, respectively.
A typical type and effect system would have to reject this
program, but this would be very restrictive in our case.
We resolve this issue by requiring that the \emph{overall}
effect of both alternatives is the same.  This (very roughly)
means that, after the plus and minus signs cancel each other
out, we have equal numbers of plus or minus signs for each
lock in both alternatives.
Furthermore, we assign the \emph{combined} effect of the
two alternatives to the conditional statement, thus keeping
track of the effect of both branches;
in the example above, the combined effect is denoted by
$[x+\then x-] \tfsep [y+\then y-]$.

The third and most complicated issue that we need to face
is support for recursive functions.
Again, consider a simple recursive function of the form
\begin{ndisplay}
  \cfont{fix}\, f.
  \ \lambda\, x.
  \ \cfont{if}\ c
  \ \cfont{then}\ (...\ f(y)\ ...)
  \ \cfont{else}\ ...
\end{ndisplay}
Let us call $\gamma_f$ the effect of $f$ and $\gamma_b$ the
computed effect for the body of $f$.
It is easy to see that $\gamma_b$ must \emph{contain} $\gamma_f$
and, if any lock/unlock operations are present in the body of $f$,
$\gamma_b$ will be strictly larger than $\gamma_f$.
Again, a typical type and effect system would require that
$\gamma_b = \gamma_f$ and reject this function definition.
We resolve this issue by computing a \emph{summary} of $\gamma_b$
and requiring that the summary is equal to $\gamma_f$.
In computing the summary, we can make several simplifications
that preserve the calculation of future locksets for operations
residing \emph{outside} function $f$.
For instance, we are not interested whether a lock is acquired
and released many times or just once, we are not interested in
the exact order in which lock/unlock pairs occur, and we can
flatten branches.

\section{Formalism} \label{sec:formalism}
The syntax of our language is illustrated in \figref{fig:syntax:core},
where $x$ and $\rho$ range over term and ``region'' variables, respectively.
\newcommand{\lsdist}{.6\textwidth}
\begin{figure}[t]
  \small
  \begin{minipage}[t]{\lsdist}
    \gbox{1.65cm}{2cm}
	 {
           \grExpr
           \grVal
	   \grFunc
    	 }
  \end{minipage}
  \begin{minipage}[t]{.35\textwidth}
    \gbox{2.0cm}{2cm}
	 {
           \grType
           \grRegion
           \grScope
           \grCap
           \grEffectElt
	 }
  \end{minipage}
  \caption{Language syntax.\label{fig:syntax:core}}
\end{figure}%
Similarly to our previous work~\cite{Cycreglock@TLDI-10},
a region is thought of as a memory unit that can be shared between threads
and whose contents can be atomically locked.
In this paper, we make the simplistic assumption that there is a one-to-one
correspondence between regions and memory cells (locations), but this is of
course not necessary.

The language core comprises of
  variables ($x$),
  constants (the unit value, $\otrue$ and $\ofalse$),
  functions ($f$), and
  function application.
Functions can be location polymorphic ($\opoly{\rho}{f}$) and
location application is explicit ($e[\rho]$).
Monomorphic functions ($\lambda x.\, e$)
must be annotated with their type.
The application of monomorphic functions is annotated with a
\emph{calling mode} ($\xi$),
which is $\nL{\gamma}$ for normal (sequential) application
     and $\nP{\epsilon}$ for parallel application.
Notice that sequential application terms are annotated with $\gamma$,
the \emph{continuation effect} as mentioned earlier.
The semantics of parallel application is that, once the
parameters have been evaluated and substituted,
the function's body is moved to a new thread of execution and
the spawning thread can proceed with the remaining computation
in parallel with the new thread.
The term $\oeval{\gamma}{e}$ encloses a function body $e$ and
can only appear during evaluation. 
The same applies to constant locations $\imath @n$, which cannot exist
at the source-level.
The construct $\olet{\gamma}{\rho}{x}{e_1}{e_2}$ allocates a fresh cell,
initializes it to $e_1$, and associates it
with variables $\rho$ and $x$ within expression $e_2$.
As in other approaches, we use $\rho$ as
the type-level representation of the new cell's location.
The reference variable $x$ has the singleton type $\tref{\rho}{\tau}$,
where $\tau$ is the type of the cell's contents.
This allows the type system to connect $x$ and $\rho$ and
thus to statically track uses of the new cell.
As will be explained later, the cell can be consumed either by
deallocation or by transferring its ownership to another thread.
Assignment and dereference operators are standard.
The value $\oloc{\imath}$ represents a reference to a location $\imath$
and is introduced during evaluation. Source programs cannot contain
$\oloc{\imath}$.

At any given program point, each cell is associated with a
\emph{capability} ($\kappa$).
Capabilities consist of two natural numbers, the \emph{capability counts}:
the \emph{cell reference} count,
which denotes whether the cell is live,
and the \emph{lock} count,
which denotes whether the cell
has been locked to provide
the current thread with
exclusive access to its contents.
Capability counts determine the validity of operations on cells.
When first allocated, a cell starts with capability $(1, 1)$,
meaning that it is live and locked, which provides exclusive
access to the thread which allocated it.
(This is our equivalent of thread-local data.)
Capabilities can be either \emph{pure} ($n_1,n_2$) or \emph{impure}
$(\overline{n_1,n_2})$. In both cases, it is implied that the
current thread can decrement the cell reference count $n_1$ times
and the lock count $n_2$ times.
Similarly to \emph{fractional permissions}~\cite{Boyland@SAS-03},
impure capabilities denote that a location may be aliased.
Our type system requires aliasing information so as to determine
whether it is safe to pass lock capabilities to new threads.

The remaining language constructs ($\oshare{}{e}$, $\orelease{}{e}$,
$\olock{\gamma}{e}$ and $\ounlock{}{e}$) operate on a reference $e$.
The first two constructs \emph{increment} and \emph{decrement} the
cell reference count of $e$ respectively. Similarly, the latter two
constructs \emph{increment} and \emph{decrement} the lock count of $e$.
%
%
As mentioned earlier, the runtime system inspects the lock
annotation $\gamma$ to determine whether it is safe to lock $e$.
%
%
%
%
%
%
%

\subsection{Operational Semantics} \label{sec:operational}
We define a \emph{small-step} operational semantics for our language
in \figref{fig:operational}.\footnote{%
  Due to space limitations, some of the functions and judgements that
  are used by the operational and (later) the static semantics are not
  formally defined in this paper.
\ifeptcstr
  Verbal descriptions and a full formalization
  are given in the Appendix.
\else
  Verbal descriptions are given in the Appendix.
  A full formalization is given in the
  companion technical report~\cite{ReglockDeadlock}.
\fi
}
%
\newcommand{\ecdista}{0.42\textwidth}%
\newcommand{\ecdistb}{0.57\textwidth}%
\begin{figure}[t]
  \begin{minipage}[t]{\ecdista}
    \gbox{2.5cm}{2cm}
	 {
           \grConf
           \grRList
           \grThread
           \grEpsilon
           \grDynCounts
    	 }
  \end{minipage}
  \hfill
  \begin{minipage}[t]{\ecdistb}
    \gbox{1.2cm}{1cm}{
      \grEvalCont
    }
  \end{minipage}
  \par\noindent
  \bgroup\footnotesize
    \textbf{Reduction relation} \hfill \fbox{$C \oparrowt{} C'$}
  \egroup
  \vspace*{-6pt}\par\noindent
  \bgroup\centering
  \newcommand\mysmall{\fontsize{8}{9.5}\selectfont}%
  \begin{nruledisplay}\footnotesize
    \orSpawn \rspace \orTerminate \rline
    \orApp \rspace \orEval \rline
    \orRPoly[break] \rspace \orFix[break] \rline[2.5em]
    \orIT[break] \rspace \orIF[break] \rline
    \orNewReg \rline
    \orAsgn \rspace \orDeref \rline
    \orShare \rspace \orRelease \rline
    \orLockZero \rline
    \orLockOne \rspace \orUnlock
  \end{nruledisplay}
  \egroup
  \vskip -6pt
  \caption{Operational semantics.%
    \label{fig:operational}}
\end{figure}%
%
The evaluation relation transforms \emph{configurations}.
A configuration $C$
consists of an abstract \emph{store} $S$ and
a thread map $T$.%
\footnote{The order of elements in comma-separated lists,
  e.g.,\ in a store $S$ or in a list of threads $T$,
  is unimportant; we consider
  all list permutations as equivalent.}
A store $S$ maps constant locations ($\imath$) to values ($v$).
%
%
A thread map $T$ associates thread identifiers to expressions
(i.e.,\ threads) and access lists.
An \emph{access list} $\theta$ maps location identifiers to
\emph{reference} and \emph{lock} counts.

A \emph{frame} $F$ is an expression
with a \emph{hole}, represented as $\opbox$.
The hole indicates the position where the next reduction step
can take place.
A \emph{thread evaluation context} $E$,
is defined as a stack of nested frames.
Our notion of evaluation context imposes a call-by-value
evaluation strategy to our language.
Subexpressions are evaluated in a left-to-right order.
We assume that concurrent reduction events can be totally
ordered~\cite{Lamport@TOPLAS-79}.
At each step, a \emph{random} thread ($n$) is chosen from the thread
list for evaluation.
Therefore, the evaluation rules are \emph{non-deterministic}.

When a parallel function
application redex is detected within the evaluation context of a
thread, a new thread is created (rule \nrulelabel{E-SN}).
The redex is replaced with a unit value in the currently executed thread
and a new thread is added to the thread list, with a \emph{fresh}
thread identifier. The calling mode of the application term is
changed from parallel to sequential.
The continuation effect associated with the sequential annotation
equals the resulting effect of the function being applied
(i.e.,\ $\fmin{\gamma_a}$).
Notice, that $\theta$ is divided into two lists
$\theta_1$ and $\theta_2$ using the new thread's initial
effect $\fmax{\gamma_a}$ as a reference for consuming
the appropriate number of counts from $\theta$.
On the other hand, when evaluation of a thread reduces to a unit
value, the thread is removed from the thread list (rule \nrulelabel{E-T}).
This is successfuly only if the thread has previously released all
of its resources.

The rule for sequential function application (\nrulelabel{E-A})
reduces an application redex to a \cfont{pop} expression, which contains
the body of the function and is annotated with the same effect as
the application term.
Evaluation propagates through \cfont{pop} expressions
(rule \nrulelabel{E-PP}), which are only useful for calculating
future locksets in rule \nrulelabel{E-LK0}.
The rules for evaluating the application of polymorphic functions
(\nrulelabel{E-RP}) and recursive functions (\nrulelabel{E-FX})
are standard, as well as the rules for evaluating conditionals
(\nrulelabel{E-IT} and \nrulelabel{E-IF}).

The rules for reference allocation, assignment and dereference
are straightforward.
Rule~\nrulelabel{E-NG} appends a fresh location $\imath$
(with initial value $v$) and the dynamic count $(1,1)$
to $S$ and $\theta$ respectively.
Rules~\nrulelabel{E-AS} and~\nrulelabel{E-D} require that
the location ($\imath$) being accessed is both live and accessible
and no other thread has access to $\imath$.
Therefore dangling memory location accesses as well as unsynchronized
accesses cause the evaluation to get \emph{stuck}.
Furthermore, the rules~\nrulelabel{E-SH}, \nrulelabel{E-RL}
and~\nrulelabel{E-UL} manipulate a cell's reference or lock
count.
They are also straightforward, simply checking that the cell
is live and (in the case of \nrulelabel{E-UL}) locked.
Rule \nrulelabel{E-RL} makes sure that a cell is unlocked
before its reference count can be decremented to zero.

The most interesting rule is~\nrulelabel{E-LK0}, which applies
when the reference being locked ($\imath$) is initially unlocked.
The future lockset ($\epsilon$) is dynamically computed,
by inspecting the preceding stack frames (${E}$)
as well as the lock annotation ($\gamma_1$).
The lockset $\epsilon$ is a list of locations (and thus locks).
The reference $\imath$ must be live
and no other thread must hold either $\imath$ or any of the
locations in $\epsilon$.
Upon success, the lock count of $\imath$ is incremented by one.
On the other hand, rule~\nrulelabel{E-LK1} applies when $\imath$
has already been locked by the current thread
(that tries to lock it again).
This immediately succeeds and the lock count is incremented by one.

\subsection{Static Semantics} \label{sec:typing}
We now present our type and effect system
and discuss the most interesting parts.
Effects are used to statically track the capability of each cell.
An effect ($\gamma$) is an \emph{ordered list} of elements of the form
$\loc^\kappa$ and summarizes the sequence of operations
(e.g., locking or sharing) on references.
%
%
The syntax of types in \figref{fig:syntax:core}
(on page~\pageref{fig:syntax:core})
is more or less standard:
Atomic types consist of base types
(the unit type, denoted by $\tunit$, and $\tbool$);
reference types $\tref{\tau}{\loc}$
are associated with a type-level cell name $\loc$
and monomorphic function types carry an \emph{effect}.
\Figref{fig:typrules} contains the typing rules.
The typing relation is denoted by
$\ttype{\cstda}{e}{\tau}{\gamma}{\gamma'}$, where
$\cstda$ is the typing context,
$e$ is an expression,
$\tau$ is the type attributed to $e$,
$\gamma$ is the \emph{input effect}, and
$\gamma'$ is the \emph{output effect}.
In the typing context,
$M$ is a mapping of constant locations to types,
$\Delta$ is a set of cell variables, and
$\Gamma$ is a mapping of term variables to types.
\begin{figure*}[t]
  \bgroup\footnotesize
    \textbf{Typing relation} \hfill
    \fbox{$\ttype{\cstda}{e}{\tau}{\gamma}{\gamma'}$}
  \egroup
  \bgroup\centering
  \begin{nruledisplay}\footnotesize
    \trUnit \rspace\hskip -2pt \trTrue \rspace\hskip -2pt \trFalse \rline
    \trVar \rspace
    \trFunc \rline
    \trRPFunc \rspace
    \trRPApp[long] \rline
    \trApp \rspace
    \trEval \rline
    \trFX \rspace
    \trLoc \rline
    \trNewRgn[long] \rline
    \trAsgn \rspace
    \trDeref \rline
    \trShare \rspace
    \trRelease \rline
    \trLock \rspace
    \trUnlock \rline
    \trIF[long]
  \end{nruledisplay}
  \egroup
  \vskip -6pt
  \caption{Typing rules.\label{fig:typrules}}
\end{figure*}%

Lock operations and sequential application terms are annotated
with the continuation effect.
This imposes the restriction that effects must flow backwards.
The input effect $\gamma$ to an expression $e$ is indeed the
continuation effect; it represents the operations that follow
the evaluation of $e$.
On the other hand, the output effect $\gamma'$ represents the
combined operations of $e$ and its continuation.
The typing relation guarantees that the input effect is
always a \emph{prefix} of the output effect.

The typing rules \nrulelabel{T-U}, \nrulelabel{T-TR},
\nrulelabel{T-FL}, \nrulelabel{T-V}, \nrulelabel{T-L},
\nrulelabel{T-RF} and \nrulelabel{T-RP} are almost
standard, except for the occasional premise $\tau \simeq \tau'$
which allows the type system to ignore the identifiers
used for location aliasing and, for example, treat the types
$\imath @{n_1}$ and $\imath @{n_2}$ as equal.
The typing rule \nrulelabel{T-F} checks that, if the effect
$\gamma_b$ that is annotated in the function's type is well formed,
it is indeed the effect of the function's body.
On the other hand, the typing rule \nrulelabel{T-A} for function
application has a lot more work to do.
It joins the input effect $\gamma$ (i.e., the continuation effect)
and the function's effect $\gamma_a$, which contains the entire
history of events occurring in the function body;
this is performed by the premise
$\xi \vdash \gamma_2 = \gamma \gplus \gamma_a$, which performs
all the necessary checks to ensure that all the capabilities
required in the function's effect $\gamma_a$ are available,
that pure capabilities are not aliased, and, in the case of
parallel application, that no lock capabilities are split
and that the resulting capability of each location is zero.
Rule \nrulelabel{T-PP} works as a bridge between the body of a
function that is being executed and its calling environment.
Rule \nrulelabel{T-FX}
uses the function $\mathsf{summary}$ to summarize the effect of
the function's body and to check that the type annotation indeed
contains the right summary.
The effect summary is conservatively computed as
the set of locks that are acquired within the function body; the
unmatched lock/unlock operations are also taken into account.

Rule \nrulelabel{T-NG} for creating new cells passes the input
effect $\gamma$ to $e_2$, the body of $\cfont{let}$, augmented
by $\rho^{0,0}$.
This means that, upon termination of $e_2$, both references and
locks of $\rho$ must have been consumed.
The output effect of $e_2$ is a $\gamma_1$ such that $\rho$ has
capability $(1,1)$, which implies that when $e_2$ starts being
evaluated $\rho$ is live and locked.
The input effect of the cell initializer expression $e_1$ is equal
to the output effect of $e_2$ without any occurrences of $\rho$.
Rules \nrulelabel{T-AS} and \nrulelabel{T-D} check that, before
dereferencing or assigning to cells, a capability of at least
$(1, 1)$ is held.
Rules~\nrulelabel{T-SH}, \nrulelabel{T-RL}, \nrulelabel{T-LK} and
\nrulelabel{T-UL} are the ones that modify cell capabilities.
In each rule, $\kappa$ is the capability after the operation
has been executed.
In the case of \nrulelabel{T-RL}, if the reference count for
a cell is decremented to zero, then all locks must have previously
been released.
The last rule in \figref{fig:typrules}, and probably the least
intuitive, is \nrulelabel{T-IF}.
Suppose $\gamma$ is the input (continuation) effect to a conditional
expression.
Then $\gamma$ is passed as the input effect to both branches.
We know that the outputs of both branches will have $\gamma$ as
a common prefix; if $\gamma_2$ and $\gamma_3$ are the suffixes,
respectively, then $\gamma_2 \tfsep \gamma_3$ is the combined
suffix, which is passed as the input effect to the condition $e_1$.

\section{Type Safety} \label{sec:safety}
In this section we present proof sketches for the fundamental theorems
that prove type safety of our language.%
\footnote{%
  The complete proofs are given in the
\ifeptcstr
  Appendix.
\else
  companion technical report~\cite{ReglockDeadlock}.
\fi
}
The type safety formulation is based on
proving \emph{progress}, \emph{deadlock freedom} 
and \emph{preservation} lemmata.
Informally, a program written in our language is safe
when for each thread of execution
either an evaluation step can be performed, or
the thread is waiting to acquire a lock (\emph{blocked}).
In addition, there must not exist any threads that have 
reached a deadlocked state. 
As discussed in Section~\ref{sec:operational}, a thread may 
become stuck when
it performs an illegal operation,
or when it references a location that has been deallocated,
or when it accesses a location that has not been locked. 


\begin{definition}[Thread Typing.]\hskip 1ex
Let $E[e]$ be the body of a thread
and let $\theta$ be the thread's \emph{access list}.
Thread typing is defined by the rule:
\begin{nruledisplay}
  \trThread
\end{nruledisplay}
\end{definition}
First of all, thread typing implies the typing of $E[e]$.

Secondly, thread typing establishes 
an exact correspondence between
counts of the access list $\theta$ 
and counts of \cfont{pop} expression annotations that
reside in the evaluation context $E[\oeval{\gamma_b}{\opbox}]$
(i.e.,\, $\mathsf{counts\_ok}\linebreak[3]%
  (E[\oeval{\gamma_b}{\opbox}],\theta)$).
The typing derivations
of $e$ and $E$ establish an exact correspondence between
the annotations of \cfont{pop} expressions and static effects.
Therefore, for each location $\imath$ in $\theta$, 
the dynamic reference and lock counts of $\imath$
are identical to the static counts of $\imath$
deduced by the type system. 

Thirdly, thread typing enforces the invariant that
the future lockset of an acquired lock 
at any program point
is \emph{always} a subset of the future lockset 
computed when the lock was initially acquired 
(i.e.,\, $\mfun{lockset\_ok}
               {E[\oeval{\gamma_b}{\opbox}],\theta}$).
This invariant is essential 
for establishing deadlock freedom.
Finally, all locations must be deallocated and released
when a thread terminates
($\forall r^\kappa \in \gamma_1.\, \kappa = (0,0)$). 

\begin{definition}[Process Typing.]\hskip 1ex
A collection of threads $T$ is well typed if each thread
in $T$ is well typed and thread identifiers are distinct:
\begin{nruledisplay}
  \trThreads
\end{nruledisplay}
\end{definition}

\begin{definition}[Store Typing.]\hskip 1ex
A store $S$ is well typed if there is a one-to-one correspondence
between $S$ and $M$ and all stored values are closed and well typed:
\begin{nruledisplay}
  \trStore
\end{nruledisplay}
\end{definition}

\begin{definition}[Configuration Typing.]\hskip 1ex
A configuration $S;T$ is \emph{well typed}
when both $T$ and $S$ are well typed, and
locks are acquired by at most one thread
(i.e.,\ $\mfun{mutex}{T}$ holds).
\begin{nruledisplay}
  \trConfig
\end{nruledisplay}
\end{definition}

\begin{definition}[Deadlocked State.]\hskip 1ex
\label{def:deadlockstate}
A set of threads $n_0, \ldots, n_k$, where $k > 0$,
has reached a \emph{deadlocked state},
when each thread $n_\imath$ has acquired lock
$\ell_{(\imath+1) \, \cfont{mod} \, (k+1)}$ and is waiting for lock
$\ell_\imath$.
\end{definition}

\begin{definition}[Not Stuck.]\hskip 1ex
A configuration $S;T$ is \emph{not stuck} when
each thread in $T$ can take one of the evaluation steps
in \figref{fig:operational}
or it is trying to acquire a lock which
(either itself or its future lockset)
is unavailable
(i.e.,\, $\mfun{blocked}{T,n}$ holds).
\end{definition}

Given these definitions, we can now present the main results
of this paper.
\emph{Progress},
\emph{deadlock freedom} 
and \emph{preservation}
are formalized at the \emph{program} level,
i.e., for all concurrently executed threads.

\begin{lemma}[Deadlock Freedom]\hskip 1ex
  \label{lemma:deadlockfreedom}
  If the initial configuration takes $n$ steps,
  where each step is well typed,
  then the resulting configuration has not 
  reached a deadlocked state.
\end{lemma}
\begin{proof}
\mycomment{
The assumptions imply that
   $\emptyset;\oth{0}{e}
    \oparrowts{n} 
    S_n;T_n$ 
	and
	$\forall \imath \in [0,n].
	 \exists M_\imath. M_\imath \vdash S_\imath;T_\imath$.
  Assume that $\mfun{deadlocked}{T_x}$ holds for some $x \in [0,n]$ and
 the first deadlock occuring in the program is in $T_x$ 
	(i.e. $\forall \imath.
				\imath  < x
				\Rightarrow
				\neg\mfun{deadlocked}{T_\imath}$). 
  Then, the following hold:
  \renewcommand{\labelitemi}{-}
  \begin{itemize}
	 \item $T_x=T,\oth{n_0}
		  {\theta_0;E_0[\olock{\gamma_0}
                   {\oloc{\imath_0}}]},
                  \ldots
                \oth{n_z}
                    {\theta_z;
                     E_z[\olock{\gamma_z}{\oloc{\imath_z}}]}$,
where threads $0$ to $z$ are in a deadlocked state.

\item	 $z > 0$ and
	  $\forall m_1 \in [0,z].\  
          \mfun{\theta_{m_1}}
               {\imath_{\mfun{succ}{m_1}}} > 0$, 
          where
         $succ(n) = (n + 1) \ \mathrm{mod} (z+1)$.
\end{itemize} 
}
  Let us assume that $z$ threads have reached a 
  deadlocked state and let
  $m \in [0,z-1]$, 
  $k = (m + 1) \ \mathrm{mod} \, z$ and
  $o = (k + 1) \ \mathrm{mod} \, z$.
  According
  to definition of \emph{deadlocked state}, 
  thread $m$ acquires lock $\imath_{k}$ and
  waits for lock $\imath_{m}$, whereas thread
  $k$ acquires lock $\imath_{o}$ and waits for
  lock $\imath_{k}$.
  Assume that $m$ is the first of the $z$
  threads that acquires a lock so it
  acquires lock $\imath_k$, 
  before thread $k$ acquires lock $\imath_{o}$.

  Let us assume that $S_y;T_y$ is 
  the configuration once $\imath_o$ is acquired %
  by thread ${k}$ for the first time,
  $\epsilon_{1y}$ is the corresponding lockset
  of $\imath_{o}$
  ($\epsilon_{1y} = 
      \flockset
        {%
         \imath_{o}}{1}{
         E[\oeval{\gamma_y}{\opbox}]
        }
  $)
  and 
  $\epsilon_{2y}$ is the set of all heap locations 
  ($\epsilon_{2y}= \mfun{dom}{S_y}$)
  at the time $\imath_{o}$ is acquired.
%
  %
  Then, $\imath_k$ does not belong to 
  $\epsilon_{1y}$, otherwise 
  thread ${k}$ would have been blocked 
  at the lock request of $\imath_{o}$
  as $\imath_k$ is already owned by thread $m$. 

  Let us assume that when thread $k$ attempts to
  acquire $\imath_k$, the configuration is of
  the form  $S_x;T_x$.
  According to the assumption of this lemma that all 
  configurations are well typed
  so $S_x;T_x$ is well-typed as well.
  By inversion of the typing derivation of $S_x;T_x$,
  we obtain the typing derivation  of thread  
   $n_k:\theta_k;E_k[\olock{\gamma_k'}{\oloc{\imath_k}}]$:
   $\olock{\gamma_k'}{\oloc{\imath_k}}$ is well-typed with
   input-output effect 
   $(\gamma_k';\gamma_k'')$, where 
   $\kappa = \gamma_k'(\imath_k @ n')$, $\kappa \geq (1,1)$,  
   $\gamma''=\gamma_k',(\imath_k @ n')^{\kappa - (1,0)}$,
  and 
   $\mfun{lockset\_ok}{E_k[\oeval{\gamma_k''}{\opbox}],\theta_k}$ holds,
  where $\theta_k$ is the access list of thread $k$.
  %
  %
  %
  %
  $\mfun{lockset\_ok}{E_k[\oeval{\gamma_k''}{\opbox}],\theta_k}$ implies 
   $\flockset{\imath_{o}}{n_2}{E_k[\oeval{\gamma_k''}{\opbox}]} 
     \cap
     \epsilon_1 \subseteq \epsilon_2$, where 
      $\theta_k = \theta_k',
       \imath_{o} \mapsto n_1;n_2;\epsilon_1;\epsilon_2$
  (notice that
         $n_2$ is positive, 
         $\epsilon_2 = \epsilon_{1y}$ and
         $\epsilon_1 = \epsilon_{2y}$ --- 
         this is immediate by the 
         operational steps from $S_y;T_y$ to $S_x;T_x$
         and rule~\nrulelabel{E-LK0}).

    We have assumed that 
    $m$ is the first thread to lock $\imath_k$ 
    at some step before $S_y;T_y$, thus
    $\imath_k \in \mfun{dom}{S_{y}}$ 
    (the store can only grow --- 
      this is immediate by observing the 
     operational semantics rules).
    %
    %
%
%
%
%
%
     %
%
%
 By the definition of $lockset$ function and the definition
 of $\gamma_k''$ we have that
 $\imath_k \in 
  \flockset{\imath_{o}}{n_2}{E_k[\oeval{\gamma_k''}{\opbox}]}$.
 Therefore, 
 $\imath_k \in 
 \flockset{\imath_{o}}{n_2}{E_k[\oeval{\gamma_k''}{\opbox}]} 
 \cap
 \mfun{dom}{S_y} 
 \subseteq
 \epsilon_{1y}$,
 which is a contradiction.
 %
\end{proof}

\begin{lemma}[Progress]\hskip 1ex
  \label{lemma:thread_progress}
  If $S;T$ is a well typed configuration, then
  $S;T$ is not stuck.
\end{lemma}
\begin{proof}
  It suffices to show that for any thread in $T$, a step can
  be performed or $block$ predicate holds for it. 
  Let $n$ be an arbitrary thread in $T$ such that   %
  $T=T_1,\oth{n}{\theta;e}$ for some $T_1$.
  By inversion of the typing derivation of $S;T$ 
  we have that
  $M;\emptyset;\emptyset \vdash_{t} \theta;e : 
   \tunit \teff{\gamma}{\gamma'}$,  
  $\mfun{mutex}{T}$, and $M \vdash S$.

  If $e$ is a \emph{value} then by inversion of
      $\ttypeext{M;\emptyset;\emptyset}
      	        {\theta;e}
	        {\tunit}{\gamma}{\gamma'}$, we obtain
   that 
   $\gamma = \gamma'$, 
   $E[e] = \opbox[\ounit]$
   and
   $\forall \imath. \mfun{\theta}{\imath} = (0,0)$,
   as a consequence of 
    $\forall r^\kappa \in \gamma. \kappa = (0,0)$ and
    $\mfun{counts\_ok}{\opbox[\oeval{\gamma}{\opbox}],\theta}$.
   Thus, rule~\nrulelabel{E-T} can be applied. 

   If $e$ is not a value then it can be trivially
   shown (by induction on the 
   typing derivation of $e$) that
   there exists a redex $u$ and an evaluation context $E$
   such that $e=E[u]$. 
   By inversion of the thread typing derivation for
   $e$ we obtain that
   $M;\emptyset;\emptyset\vdash u : 
     \tau \teff{\gamma_a}{\gamma_b}$,
   $M;\emptyset;\emptyset \vdash E : 
     \tefunc{\tau}{\gamma_a}{\gamma_b}{\tunit} 
     \teff{\gamma}{\gamma'}$,
   $\mfun{counts\_ok}{E[\oeval{\gamma_b}{\opbox}],\theta}$ hold. 

   Then, we proceed by perfoming a case analysis on $u$
   (we only consider the most interesting cases): 
   \casesplit{

   \item  $\oapp{\ofunc{\gamma}{x}{e'}{\tau}}{v}{\nP{}}$:
        it suffices to show that
        $\fsplit{\theta_1}{\theta_2}{\theta}{\fmax{\gamma_c}}$ is defined,
        where $\gamma_c$ is
        the nnotation of type $\tau$.
      %
      %
      If $\fmax{\gamma_c}$ is empty, 
      then the proof is immediate from the base case of
      $\mathsf{split}$ function.
      Otherwise, we must show that for all $\imath$, 
      the count $\theta(\imath)$ is greater than or equal
      to the sum of all 
      $(\imath @ n)^{\kappa}$ in $\fmax{\gamma_c}$.
      This can be shown by considering 
      $\nP{} \vdash 
        \gamma_b = \gamma_a \gplus \gamma_c$ 
      (i.e.,\, the $max$ counts in $\gamma_c$ are less than
           or equal to the $max$ counts in $\gamma_b$), which can be obtained by inversion
   	   of the typing derivation of $\oapp{\ofunc{\gamma}{x}{e'}{\tau}}{v}{\nP{}}$, 
           and the exact correspondence between static ($\gamma_b$) and dynamic counts 
            (i.e,\, %
                    $\mfun{counts\_ok}{E[\oeval{\gamma_b}{\opbox}],\theta}$).
      Thus, rule~\nrulelabel{E-SN} can be applied 
      to perform a single step.

       \item  $ \oshare{\gamma}{\oloc{\imath}}$: 
	$\mfun{counts\_ok}{E[\oeval{\gamma_b}{\opbox}],\theta}$ 
        establishes an exact correspondence 
        between dynamic and static counts. The typing derivation implies that
	$\gamma_a(\imath @ n_1) \geq (2,0)$, for some $n_1$ existentially bound
	in the premise	of the derivation.
	Therefore, $\theta(\imath) \geq (1,0)$.
	It is possible to perform a single step using rule~\nrulelabel{E-SH}.
              The cases for 
		$\orelease{\gamma}{\oloc{\imath}}$ and
                $ \ounlock{\gamma}{\oloc{\imath}}$ 
              can be shown in a similar manner.

      \item  $ \olock{\gamma_a}{\oloc{\imath}}$: 
        similarly to the case we can show
	that  $\theta(\imath) = (n_1,n_2)$ and $n_1$ is positive.
        If $n_2$ is positive, rule~\nrulelabel{E-LK1} can be applied.
        Otherwise, $n_2$ is zero. Let $\epsilon$ be
        equal to 
        $ \mfun{locked}{T_1} \cap 
          \flockset{\imath}{1}{{E[\oeval{\gamma_a}{\opbox}]}}
        $.
        If $\epsilon$ is empty then rule~\nrulelabel{E-LK0} can be applied
        in order to perform a single step.
        Otherwise, $\mfun{blocked}{T,n}$ predicate holds and the
        configuration is not stuck.

			%
                        %
                        %

	\item  $\oderef{\oloc{\imath}}$:
				it can be trivially shown (as in the previous case of $share$
				that we proved $\theta(\imath) \geq (1,0)$), that
				 $\theta(\imath) \geq (1,1)$ and since 
					$\mfun{mutex}{T_1,\oth{n}{\theta;E[\oderef{\oloc{\imath}}]}}$ holds,
				then $\imath \notin \mfun{locked}{T_1}$ and thus rule~\nrulelabel{E-D}
				can be used to perform a step.
                 The case of $\oassign{\oloc{\imath}}{v}$ can be 
                 shown in a similar manner.
                 \qedhere
   }
\end{proof} 

%
\begin{lemma}[Preservation]\hskip 1ex
  \label{lemma:thread_preservation}
  Let $S;T$ be a well-typed configuration with ${M} \vdash {S;T}$.
  If the operational semantics takes a step
  $S;T \oparrowt{n} S';T'$,
  then there exists $M' \supseteq M$
  such that the resulting configuration is well-typed
  with ${M'} \vdash {S';T'}$.
\end{lemma}
\begin{proof}
 We proceed by case analysis on 
 the thread evaluation relation 
 (we only consider a few cases due to space limitations):
 \casesplit{
      \item \nrulelabel{E-A}: %
	 Rule~\nrulelabel{E-A} implies
	 $S'=S$, 
	$\ \ T'= T,\oth{n}{\theta;E[\oeval{\gamma_a}{e_1[v/x]}]}$ 
				and 
	$e = \oapp{\ofunc{}{x}{e_1}
                {\tfunc{\tau_1}{\gamma_c}{}{\tau_2}}}{v}{\nL{\gamma_a}}$.
	 By inversion of the 
         configuration typing assumption  we have that
         $\mfun{mutex}{T,\oth{n}{\theta;E[e]}}$
        and
	$\ttypeext{M;\emptyset;\emptyset}
                  {\theta;E[e]}
         	  {\tunit}{\gamma}{\gamma'}$ hold.
        It suffices to show that
        $\mfun{mutex}{T,
	             \oth{n}{\theta;
                     E[\oeval{\gamma_a}{e_1[v/x]}]}}$
        and
	$\ttypeext{M;\emptyset;\emptyset}
                  {\theta;E[\oeval{\gamma_a}{e_1[v/x]}]}
         	  {\tunit}{\gamma}{\gamma'}$ hold.
        The former is immediate from 
        $\mfun{mutex}{T,\oth{n}{\theta;E[e]}}$ as
	          no new locks are acquired. 
        Now we proceed with the latter,
        which can be shown by proving that
        $ \ttype{M;\emptyset;\emptyset}
                    {\oeval{\gamma_a}{e_1[v/x]}}
                   {\tau_2'}{\gamma_a}{\gamma_b}$ holds.
        By inversion on the thread typing derivation 
        $E[e]$ we have 
        %
        %
	$M;\emptyset;\emptyset \vdash v  : \tau_1'
         \teff{\gamma_b}{\gamma_b}$, 
        %
         $\nL{\gamma_a} \vdash 
            \gamma_b = \gamma_a \gplus \gamma_c'$ and
   	 $M;\emptyset;\emptyset \vdash 
           \ofunc{}{x}{e_1}
                  {\tfunc{\tau_1}{\gamma_c}{???}{\tau_2}}
            :
            \tfunc{\tau_1'}{\gamma_c'}{???}{\tau_2'} 
            \teff{\gamma_b}{\gamma_b}$, 
          where 
          $\tfunc{\tau_1'}{\gamma_c'}{???}{\tau_2'}
          \simeq 
          \tfunc{\tau_1}{\gamma_c}{???}{\tau_2}$.
         %
         %
         We can use proof by induction on 
         the expression typing relation
         to show that if $v$ is well typed with $\tau_1'$, then
         it is also well typed with $\tau_1$ provided that
         $\tau_1 \simeq \tau_1'$. Therefore, 
 	     $M;\emptyset;\emptyset \vdash v  : \tau_1
                \teff{\gamma_b}{\gamma_b}$ holds.
         By inversion of the function typing derivation 
         we obtain that
         $\nL{\emptyset} \vdash \gamma_c
          \Rightarrow
           \ttype{M;\emptyset;\emptyset,x:\tau_1}
                {e_1}{\tau_2}{\fmin{\gamma_c}}{\gamma_c}$.
             $\nL{\emptyset} \vdash \gamma_c'$ 
         (premise of 
            $\nL{\gamma_a} \vdash 
            \gamma_b = \gamma_a \gplus \gamma_c'$)
         and $\gamma_c \simeq \gamma_c'$
         imply that
          $\nL{\emptyset} \vdash \gamma_c$ holds, thus
         $   \ttype{M;\emptyset;\emptyset,x:\tau_1}
                   {e_1}{\tau_2}{\fmin{\gamma_c}}{\gamma_c}$ 
         holds.
         %
         %
         By applying the standard value substitution lemma on
         the new typing derivation of $v$ we obtain that
         $\ttype{M;\emptyset;\emptyset}
                   {e_1[v/x]}{\tau_2}
                   {\fmin{\gamma_c}}{\gamma_c}$ holds.
            The application of rule~\nrulelabel{T-PP} 
            implies that
            $ \ttype{M;\emptyset;\emptyset}
              {\oeval{\gamma_a}{e_1[v/x]}}
             {\tau_2'}{\gamma_a}{\gamma_b}$ holds.

\item \nrulelabel{E-LK0}, \nrulelabel{E-LK1}, \nrulelabel{E-UL},
       \nrulelabel{E-SH} and \nrulelabel{E-RL}:
these rules generate side-effects as they modify 
the reference/lock count of location $\imath$.
We provide a single proof for all cases. 
Hence, we are assuming 
here that $u$ (i.e.\, in $E[u]$) has one of the following
forms:
         $\olock{\gamma_1}{\oloc{\imath}}$,
         $\ounlock{\gamma_1}{\oloc{\imath}}$
         $\oshare{\gamma_1}{\oloc{\imath}}$ or 
         $\orelease{\gamma_1}{\oloc{\imath}}$.
Rules~\nrulelabel{E-LK0}, \nrulelabel{E-LK1},
      \nrulelabel{E-UL}, \nrulelabel{E-SH} 
      and \nrulelabel{E-RL} imply that
$S'=S$, 
$\ \ T'= T,\oth{n}{\theta';E[\ounit]}$, 
where $\ounit$ replaces $u$ in context $E$
and $\theta$ differs with respect to $\theta'$
only in the one of the counts of $\imath$
(i.e.,\,
$\theta' = 
 \theta[\imath \mapsto
 \mfun{\theta}{\imath}+(n_1,n_2)]$ and
 $\gamma_a(r) - \kappa = (n_1,n_2)$ --- $\gamma_a$ is the 
 input effect of $E[u]$). 

By inversion of the configuration typing assumption  we have that:
\listitems{
\item $\mfun{mutex}{T,\oth{n}{\theta;E[u]}}$: 
In the case of
\nrulelabel{E-UL},
\nrulelabel{E-SH}, \nrulelabel{E-LK1} 
and \nrulelabel{E-RL} no new locks are acquired. 
Thus, $\mfun{mutex}{T,\oth{n}{\theta';E[\ounit]}}$ holds.
In the case of rule~\nrulelabel{E-LK0}, 
a new lock $\imath$ is acquired
(i.e.,\, when the lock count of $\imath$ is zero) 
the precondition 
of~\nrulelabel{E-LK0} suggests that 
no other thread holds $\imath$:
$\mfun{locked}{T} \cap
 \flockset{\imath}{1}
      {{E[\oeval{\gamma_a}{\opbox}]}} 
  = \emptyset$. %
%
Thus,
$\mfun{mutex}{T,\oth{n}{\theta';E[\ounit]}}$ holds.

  \item	$\ttypeext{M;\emptyset;\emptyset}
                           {\theta;E[u]}{\tunit}{\gamma}{\gamma'}$:
      By inversion %
      we have that
       $M;\emptyset;\emptyset	\vdash 
	E : \tefunc{\tunit}
             {\gamma_a}{\gamma_b}
              {\tunit} 
              \teff{\gamma}{\gamma'}$
      and
     $M;\emptyset;\emptyset 
     \vdash 
      u : \tunit \teff{\gamma_a}{\gamma_b}$,
      where
       $\gamma_b = \gamma_a,(\imath @ n')^\kappa$ for
       some $n'$.
      It can be trivially shown from the latter derivation
      that
      $M;\emptyset;\emptyset \vdash \ounit : \tunit
         \teff{\gamma_a}{\gamma_a}$. 
      We can obtain from the typing derivation of $E$ 
      (proof by induction) that
      $M;\emptyset;\emptyset	\vdash 
	E : \tefunc{\tunit}
             {\gamma_a}{\gamma_a}
             {\tunit}  
            \teff{\gamma}{\gamma''}$,
       where
       $\gamma' = \gamma'', (\imath @ n')^\kappa$.
         %
         %
         %

   \item $\mfun{lockset\_ok}{E[\oeval{\gamma_b}{\opbox}],\theta}$ 
         and $\mfun{counts\_ok}{E[\oeval{\gamma_b}{\opbox}],\theta}$:
      %
      By the definition of $lockset$ function it can be shown that
      $\flockset{\jjmath}{n_b}{E[\oeval{\gamma_a}{\opbox}]}
       \subseteq
       \flockset{\jjmath}{n_b}{E[\oeval{\gamma_b}{\opbox}]}$
      for all $\jjmath \neq \imath$ in the domain of
      $\theta'$ ($n_b$ is the lock count of $\jjmath$
      in $\theta$).
      The same applies for $\jjmath = \imath$ in the case
      of rules~\nrulelabel{E-SH}, \nrulelabel{E-RL} as the lock count
      of $\imath$ is not affected. 
      In the case of rules~\nrulelabel{E-LK0}, \nrulelabel{E-LK1},
      \nrulelabel{E-UL} we have
       $\flockset{\imath}{n_b \pm 1}{E[\oeval{\gamma_a}{\opbox}]}$, but this
      is identical to %
       $\flockset{\imath}{n_b}{E[\oeval{\gamma_b}{\opbox}]}$ 
      by the definition of $lockset$. 
      Therefore $\mfun{lockset\_ok}{E[\oeval{\gamma_a}{\opbox}],\theta'}$ holds.
      The predicate 
      $counts\_ok$ $(E[\oeval{\gamma_b}{\opbox}],\theta)$
      enforces the invariant that
      the static counts 
      are identical to the dynamic 
         counts ($\theta$) of $\imath$. 
      The lock count of $\theta$ is modified by $\pm 1$
      and $\gamma_a$ differs with respect to $\gamma_b$
      by $(\imath @ n')^\kappa$. We can use this fact
      to show that $\mfun{counts\_ok}{E[\oeval{\gamma_a}{\opbox}],\theta'}$.
      \qedhere
%
%
}
}
\end{proof} 

\begin{lemma}[Multi-step Program Preservation]\hskip 1ex
  \label{lemma:mthread_preservation}
  Let $S_0;T_0$ be a \emph{closed well-typed configuration} 
  for some $M_0$ and assume that
  $S_0;T_0$ evaluates to $S_n;T_n$ in $n$ steps.
  Then for all $\imath \in [0,n]$
  $M_\imath \vdash S_\imath;T_\imath$ holds.
\end{lemma}
\begin{proof}
 Proof by induction on the number of steps $n$ using
 Lemma~\ref{lemma:thread_preservation}.
\end{proof}

\begin{theorem}[Type Safety]\hskip 1ex
  \label{theorem:type_safety}
Let expression $e$ be the initial program
and let the initial typing context $M_0$
and the initial program configuration $S_0; T_0$ be defined as follows:
  $M_0      = \emptyset$,
  $S_0      = \emptyset$, and
  $T_0      = \{ 0 :\emptyset;e \}$.
  If $S_0;T_0$ is well-typed in $M_0$ and
  the operational semantics takes any number of steps
  $S_0;T_0 \oparrowt{n}^{n} S_n;T_n$,
  then the resulting configuration $S_n;T_n$ is not stuck
  and $T_n$ has not reached a deadlocked state.
\end{theorem}
\begin{proof}
  The application of Lemma~\ref{lemma:mthread_preservation} 
  to the typing derivation of $S_0;T_0$ implies 
  that for all steps from zero to $n$
  there exists an $M_\imath$ such that
  $M_\imath \vdash S_\imath;T_\imath$.
  Therefore, Lemma~\ref{lemma:deadlockfreedom} 
  implies that $\neg\mfun{deadlocked}{T_{n}}$ and
  Lemma~\ref{lemma:thread_progress} implies
  $S_{n};T_{n}$ is not stuck.
\end{proof}

Typing the initial configuration $S_0;T_0$
with the empty typing context $M_0$ guarantees that
all functions in the program are closed 
and that no explicit location values ($\oloc{\imath}$)
are used in the original program.

\section{Concluding Remarks} \label{sec:conc}
%
The main contribution of this work is
type-based deadlock avoidance for a
language with unstructured locking primitives
and the meta-theory for the proposed semantics.
The type system presented in this paper guarantees
that well-typed programs 
will not deadlock at execution time.
This is possible by statically verifying that
program annotations 
reflect the order of future lock operations and
using the annotations  
at execution time to avoid deadlocks.
The main advantage over purely static approaches
to deadlock freedom is that our type system 
accepts a wider class of programs as it does not
enforce a total order on lock acquisition. 
The main disadvantages of our approach is that
it imposes an additional run-time 
overhead induced 
by the future lockset computation 
and blocking time (i.e.,\, both the requested lock
and its future lockset must be available). 
Additionally, in some cases threads may unnecessarily 
block because our type and effect system is conservative.
For example, when a thread locks $x$ and executes a
lengthy computation (without acquiring other locks)
before releasing $x$, it would be safe to allow another
thread to lock $y$ even if $x$ is in its future lockset.

%

%
We have shown that this is a non-trivial extension
for existing type systems based on deadlock avoidance.
There are three significant sources of complexity:
(i)~lock acquisition and release operations may not be
    properly nested,
(ii)~lock-unlock pairs may span multiple contexts:
     function calls that contain lock operations
     may not always \emph{increase}  the size of lockset, but
     instead \emph{limit} the lockset size.
     In addition, future locksets must be computed
     in a context-sensitive manner 
     (stack traversal in our case),
 and
(iii)~in the presence of 
      location (lock) polymorphism and aliasing,
      it is very difficult for a static type system 
      even to detect the previous 
      two sources of complexity.
      To address lock aliasing without 
      imposing restrictions statically, 
      we defer lockset resolution until run-time.

\begin{small}
\section*{Acknowledgement}
This research is partially funded by the programme for supporting
basic research ($\Pi$EBE 2010) of the National Technical University of
Athens, under a project titled ``Safety properties for concurrent
programming languages.''
\end{small}

\begin{appendix}
  \section*{Appendix}
  \setcounter{section}{1}
  \trTable
\ifeptcstr
 \input{appendix_listing}
\fi
\end{appendix}


\begin{thebibliography}{10}
\providecommand{\bibitemdeclare}[2]{}
\providecommand{\urlprefix}{Available at }
\providecommand{\url}[1]{\texttt{#1}}
\providecommand{\href}[2]{\texttt{#2}}
\providecommand{\urlalt}[2]{\href{#1}{#2}}
\providecommand{\doi}[1]{doi:\urlalt{http://dx.doi.org/#1}{#1}}
\providecommand{\bibinfo}[2]{#2}

\bibitemdeclare{inproceedings}{Boudol@ICTAC-09}
\bibitem{Boudol@ICTAC-09}
\bibinfo{author}{G{\'e}rard Boudol} (\bibinfo{year}{2009}):
  \emph{\bibinfo{title}{A Deadlock-Free Semantics for Shared Memory
  Concurrency}}.
\newblock In \bibinfo{editor}{Martin Leucker} \& \bibinfo{editor}{Carroll
  Morgan}, editors: {\sl \bibinfo{booktitle}{Proceedings of the International
  Colloquium on Theoretical Aspects of Computing}}, {\sl
  \bibinfo{series}{LNCS}} \bibinfo{volume}{5684},
  \bibinfo{publisher}{Springer}, pp. \bibinfo{pages}{140--154},
  \doi{10.1007/978-3-642-03466-4\_9}.

\bibitemdeclare{inproceedings}{OwnershipTypes@OOPSLA-02}
\bibitem{OwnershipTypes@OOPSLA-02}
\bibinfo{author}{Chandrasekhar Boyapati}, \bibinfo{author}{Robert Lee} \&
  \bibinfo{author}{Martin Rinard} (\bibinfo{year}{2002}):
  \emph{\bibinfo{title}{Ownership Types for Safe Programming: Preventing Data
  Races and Deadlocks}}.
\newblock In: {\sl \bibinfo{booktitle}{Proceedings of the ACM SIGPLAN
  Conference on Object-Oriented Programming, Systems, Languages, and
  Applications}}, \bibinfo{publisher}{ACM Press}, \bibinfo{address}{New York,
  NY, USA}, pp. \bibinfo{pages}{211--230}, \doi{10.1145/582419.582440}.

\bibitemdeclare{inproceedings}{Boyland@SAS-03}
\bibitem{Boyland@SAS-03}
\bibinfo{author}{John Boyland} (\bibinfo{year}{2003}):
  \emph{\bibinfo{title}{Checking Interference with Fractional Permissions}}.
\newblock In \bibinfo{editor}{Radhia Cousot}, editor: {\sl
  \bibinfo{booktitle}{Static Analysis: Proceedings of the 10th International
  Symposium}}, {\sl \bibinfo{series}{LNCS}} \bibinfo{volume}{2694},
  \bibinfo{publisher}{Springer}, pp. \bibinfo{pages}{55--72},
  \doi{10.1007/3-540-44898-5\_4}.

\bibitemdeclare{article}{SysDeadlocks@CompSurv-71}
\bibitem{SysDeadlocks@CompSurv-71}
\bibinfo{author}{Edward~G. Coffman, Jr.}, \bibinfo{author}{Michael~J. Elphick}
  \& \bibinfo{author}{Arie Shoshani} (\bibinfo{year}{1971}):
  \emph{\bibinfo{title}{System Deadlocks}}.
\newblock {\sl \bibinfo{journal}{ACM Comput. Surv.}}
  \bibinfo{volume}{3}(\bibinfo{number}{2}), pp. \bibinfo{pages}{67--78},
  \doi{10.1145/356586.356588}.

\bibitemdeclare{inproceedings}{FlanaganAbadi@CONCUR-99}
\bibitem{FlanaganAbadi@CONCUR-99}
\bibinfo{author}{Cormac Flanagan} \& \bibinfo{author}{Mart{\'\i}n Abadi}
  (\bibinfo{year}{1999}): \emph{\bibinfo{title}{Object Types Against Races}}.
\newblock In \bibinfo{editor}{Jos C.~M. Baeten} \& \bibinfo{editor}{Sjouke
  Mauw}, editors: {\sl \bibinfo{booktitle}{International Conference on
  Concurrency Theory}}, {\sl \bibinfo{series}{LNCS}} \bibinfo{volume}{1664},
  \bibinfo{publisher}{Springer}, pp. \bibinfo{pages}{288--303},
  \doi{10.1007/3-540-48320-9\_21}.

\bibitemdeclare{inproceedings}{FlanaganAbadi@ESOP-99}
\bibitem{FlanaganAbadi@ESOP-99}
\bibinfo{author}{Cormac Flanagan} \& \bibinfo{author}{Mart{\'\i}n Abadi}
  (\bibinfo{year}{1999}): \emph{\bibinfo{title}{Types for Safe Locking}}.
\newblock In: {\sl \bibinfo{booktitle}{Programming Language and Systems:
  Proceedings of the European Symposium on Programming}}, {\sl
  \bibinfo{series}{LNCS}} \bibinfo{volume}{1576},
  \bibinfo{publisher}{Springer}, pp. \bibinfo{pages}{91--108},
  \doi{10.1007/3-540-49099-X\_7}.

\bibitemdeclare{inproceedings}{Cycreglock@TLDI-10}
\bibitem{Cycreglock@TLDI-10}
\bibinfo{author}{Prodromos Gerakios}, \bibinfo{author}{Nikolaos Papaspyrou} \&
  \bibinfo{author}{Konstantinos Sagonas} (\bibinfo{year}{2010}):
  \emph{\bibinfo{title}{Race-free and Memory-safe Multithreading: Design and
  Implementation in {Cyclone}}}.
\newblock In: {\sl \bibinfo{booktitle}{Proceedings of the ACM SIGPLAN
  International Workshop on Types in Languages Design and Implementation}},
  \bibinfo{publisher}{ACM Press}, \bibinfo{address}{New York, NY, USA}, pp.
  \bibinfo{pages}{15--26}, \doi{10.1145/1708016.1708020}.

\bibitemdeclare{techreport}{ReglockDeadlock}
\bibitem{ReglockDeadlock}
\bibinfo{author}{Prodromos Gerakios}, \bibinfo{author}{Nikolaos Papaspyrou} \&
  \bibinfo{author}{Konstantinos Sagonas} (\bibinfo{year}{2010}):
  \emph{\bibinfo{title}{A Type System for Unstructured Locking that Guarantees
  Deadlock Freedom without Imposing a Lock Ordering}}.
\newblock \bibinfo{type}{Technical Report}, \bibinfo{institution}{National
  Technical University of Athens}.
\newblock
  \urlprefix\url{http://softlab.ntua.gr/~pgerakios/papers/reglock_deadlock_tec%
hrep10.pdf}.

\bibitemdeclare{inproceedings}{DeadlockAvoidance@TLDI-11}
\bibitem{DeadlockAvoidance@TLDI-11}
\bibinfo{author}{Prodromos Gerakios}, \bibinfo{author}{Nikolaos Papaspyrou} \&
  \bibinfo{author}{Konstantinos Sagonas} (\bibinfo{year}{2011}):
  \emph{\bibinfo{title}{A Type and Effect System for Deadlock Avoidance in
  Low-level Languages}}.
\newblock In: {\sl \bibinfo{booktitle}{Proceedings of the ACM SIGPLAN
  International Workshop on Types in Languages Design and Implementation}},
  \bibinfo{publisher}{ACM Press}, \bibinfo{address}{New York, NY, USA}, pp.
  \bibinfo{pages}{15--28}, \doi{10.1145/1929553.1929558}.

\bibitemdeclare{inproceedings}{Kobayashi@CONCUR-06}
\bibitem{Kobayashi@CONCUR-06}
\bibinfo{author}{Naoki Kobayashi} (\bibinfo{year}{2006}):
  \emph{\bibinfo{title}{A New Type System for Deadlock-Free Processes}}.
\newblock In \bibinfo{editor}{C.~Baier} \& \bibinfo{editor}{H.~Hermanns},
  editors: {\sl \bibinfo{booktitle}{International Conference on Concurrency
  Theory}}, {\sl \bibinfo{series}{LNCS}} \bibinfo{volume}{4137},
  \bibinfo{publisher}{Springer}, pp. \bibinfo{pages}{233--247},
  \doi{10.1007/11817949\_16}.

\bibitemdeclare{article}{Lamport@TOPLAS-79}
\bibitem{Lamport@TOPLAS-79}
\bibinfo{author}{Leslie Lamport} (\bibinfo{year}{1979}):
  \emph{\bibinfo{title}{A New Approach to Proving the Correctness of
  Multiprocess Programs}}.
\newblock {\sl \bibinfo{journal}{ACM Transactions on Programming Languages and
  Systems}} \bibinfo{volume}{1}(\bibinfo{number}{1}), pp.
  \bibinfo{pages}{84--97}, \doi{10.1145/357062.357068}.

\bibitemdeclare{inproceedings}{NonLexicalDeadlock@APLAS-08}
\bibitem{NonLexicalDeadlock@APLAS-08}
\bibinfo{author}{Kohei Suenaga} (\bibinfo{year}{2008}):
  \emph{\bibinfo{title}{Type-Based Deadlock-Freedom Verification for
  Non-Block-Structured Lock Primitives and Mutable References}}.
\newblock In \bibinfo{editor}{G.~Ramalingam}, editor: {\sl
  \bibinfo{booktitle}{Asian Symposium on Programming Languages and Systems}},
  {\sl \bibinfo{series}{LNCS}} \bibinfo{volume}{5356},
  \bibinfo{publisher}{Springer}, pp. \bibinfo{pages}{155--170},
  \doi{10.1007/978-3-540-89330-1\_12}.

\bibitemdeclare{inproceedings}{Vasco@PLACES-09}
\bibitem{Vasco@PLACES-09}
\bibinfo{author}{Vasco Vasconcelos}, \bibinfo{author}{Francisco Martin} \&
  \bibinfo{author}{Tiago Cogumbreiro} (\bibinfo{year}{2010}):
  \emph{\bibinfo{title}{Type Inference for Deadlock Detection in a
  Multithreaded Polymorphic Typed Assembly Language}}.
\newblock In \bibinfo{editor}{Alastair~R. Beresford} \& \bibinfo{editor}{Simon
  Gay}, editors: {\sl \bibinfo{booktitle}{Proceedings of the Workshop on
  Programming Language Approaches to Concurrency and Communication-cEntric
  Software}}, {\sl \bibinfo{series}{EPTCS}}~\bibinfo{volume}{17}, pp.
  \bibinfo{pages}{95--109}, \doi{10.4204/EPTCS.17.8}.

\end{thebibliography}
\end{document}